\documentclass[11pt,onecolumn]{IEEEtran}
\usepackage{amssymb, amsmath, graphicx, cite, epsfig, booktabs}

\usepackage[usenames,dvipsnames]{pstricks}
\usepackage{epsfig}
\usepackage{pst-grad} 
\usepackage{pst-plot} 
\usepackage{bm}
\usepackage{turnstile}
\usepackage{ushort}
\usepackage{extarrows}

\newcommand{\thmend}{\hspace*{\fill}~\QEDopen\par\endtrivlist\unskip}

\newcommand{\bX}[1]{{{\mathit{\mathbf{X}}}_{#1}}}
\newcommand{\bY}{{{\mathit{\mathbf{Y}}}}}
\newcommand{\bU}{{\mathit{\mathbf{U}}}}
\newcommand{\bV}{{\mathit{\mathbf{V}}}}
\newcommand{\bZ}{{\mathit{\mathbf{Z}}}}
\newcommand{\bv}{{\mathit{\mathbf{v}}}}
\newcommand{\bbx}{{\mathit{\mathbf{x}}}}
\newcommand{\bu}{{\mathit{\mathbf{u}}}}
\newcommand{\by}{{\mathit{\mathbf{y}}}}
\newcommand{\bx}[1]{{{\mathit{\mathbf{x}}}_{#1}}}

\newcommand{\Bt}{{{\bm{\theta}}}}

\newtheorem{theorem}{\bf Theorem}
\newtheorem{lemma}{\bf Lemma}

\newtheorem{definition}{\bf Definition}
\newtheorem{corollary}{\bf Corollary}

\newtheorem{remark}{\bf Remark}

\newcommand{\qed}{\nobreak \ifvmode \relax \else
      \ifdim\lastskip<1.5em \hskip-\lastskip
      \hskip1.5em plus0em minus0.5em \fi \nobreak
      \vrule height0.35em width0.4em depth0.15em\fi}

\allowdisplaybreaks[4]

\begin{document}

\title{\vspace{-0.5cm}Separation Theorems for Phase-Incoherent Multiple-User Channels}
\author{{H.~Ebrahimzadeh Saffar, E.~Haj~Mirza~Alian, and P.~Mitran \\
Department of Electrical and Computer Engineering \\
University of Waterloo, Waterloo, Ontario, Canada\\
Email: \tt\{hamid, ealian, pmitran\}@ece.uwaterloo.ca} \vspace{-0.59cm}}



\maketitle

\begin{abstract}

We study the transmission of two correlated and memoryless sources $(U,V)$ over several
multiple-user phase asynchronous channels. Namely, we consider a class of phase-incoherent multiple
access relay channels (MARC) with both non-causal and causal unidirectional cooperation between
encoders, referred to as phase-incoherent unidirectional non-causal cooperative MARC
(PI-UNCC-MARC), and phase-incoherent unidirectional causal cooperative MARC (PI-UCC-MARC)
respectively. We also consider phase-incoherent interference channels (PI-IC), and interference
relay channel (PI-IRC) models in the same context. In all cases, the input signals are assumed to
undergo non-ergodic phase shifts due to the channel. The shifts are assumed to be unknown to the
transmitters and known to the receivers as a realistic assumption. Both necessary and sufficient
conditions in order to reliably send the correlated sources to the destinations over the considered
channels are derived. In particular, for all of the channel models, we first derive an outer bound
for reliable communication that is defined with respect to the source entropy content (i.e., the
triple $\left(H(U|V),H(V|U),H(U,V)\right)$). Then, using {\em separate} source and channel coding,
under specific gain conditions, we establish the same region as the inner bound and therefore
obtain tight conditions for reliable communication for the specific channel under study. We thus
establish a source-channel separation theorem for each channel and conclude that without the
knowledge of the phase shifts at the transmitter sides, separation is optimal. It is further
conjectured that separation in general is optimal for all channel coefficients.

\end{abstract}


\section{Introduction} 

Incoherence or asynchronism between different nodes of a communication network is an inherent
challenge to modern communication systems. In particular, there are major factors in wireless
systems, such as feedback delay, the bursty nature of some applications, and reaction delay, which
cause time or phase asynchronism between different nodes of a network
\cite{Wornell_asynchronism:2009}. Furthermore, in multi-user systems, interference from other
sources make synchronization much more difficult. Therefore, it is interesting to study multi-user
communication problems
without assuming synchronism a priori. 

In point-to-point wireless systems, achieving receiver synchronization is possible in principle,
using training sequences and/or feedback. However, although analytically convenient, full
synchronization is rarely a practical or easily justified assumption, and in some cases
theoretically infeasible \cite{Zhang_ICC:2007}.
The first studies of {\em time} asynchronism in point-to-point communications goes back to the 60's
(\cite{Chase:1968}, \cite{Dobrushin:1967}), where the receiver is not accurately aware of the exact
time that the encoded symbols were transmitted. The recent work of
\cite{Wornell_asynchronism:2009}, on the other hand, assumes a stronger form of time asynchronism,
that is, the receiver knows neither the time at which transmission starts, nor the timing of the
last information symbol. They propose a combined communication and synchronization scheme and
discuss information-theoretical limits of the model. Also, in multi-user communication settings,
the problem of time asynchronism is addressed for example in \cite{Cover_asynchronism:1981},
\cite{Verdu:1989} for the particular case of multiple access channels.

Besides time asynchronism \cite{Wornell_asynchronism:2009}, which is present in most channels,
other forms of asynchronism such as {\em phase} uncertainty are important in wireless systems. In
fading channels, the channel state information (CSI) models amplitude attenuation and phase shifts
(phase fading) introduced by the channels between the nodes. In many systems, it is difficult to
know phase shifts at the transmitter side due to the delay and resource limits in feedback
transmission. In particular, in highly mobile environments, fading in conjunction with feedback
delay may result in out of date phase knowledge by the time it reaches the transmitters (see, e.g.,
\cite{Onggosanusi_feedback_delay:2001}).

The issue of phase asynchronism can be analytically seen in the larger framework of {\em channel
uncertainty}, that is, the communicating parties have to work under situations where the full
knowledge of the law governing the channel (or channels in a multi-user setting) is not known to
some or all of them \cite{Lapidoth_Survey:1998}. In order to study this general problem from an
information-theoretic point of view, the mathematical model of a {\em compound channel} (or
state-dependent channel) has been introduced by different authors \cite{Blackwell:1959},
\cite{Csiszar:1981}, \cite{Wolfowitz:1978}. A compound channel is generally represented by a family
of transition probabilities $p^{\theta}_{Y \vert X}$, where the index $\theta \in {\Theta}$ is the
state of the channel and $\Theta$ represents the uncertainty of different parties about the exact
channel's transition probability.

In this paper, we consider the problem of joint source-channel coding for a range of compound
Gaussian multiple-user channels with phase
uncertainty and prove a separation theorem for each. 
We assume that the phase shifts over channels under consideration are stationary non-ergodic phase
fading processes which are chosen randomly and fixed over the block length. Thus, phase
asynchronism is formulated in the compound channel
framework and 
the phase information $\bm\theta$ (as the channel parameter) is assumed to be unknown to the
transmitters and known to the receiver side(s) as a practical assumption. 
Consequently, as our main contribution, we find conditions that are both necessary and sufficient
for sending a pair of correlated sources over a class of continuous alphabet multiple-user channels
under phase uncertainty.

The problem of joint source-channel coding for a network is open in general. Several works,
however, have been published on this issue for multiple access channel (MAC). As an example, for
lossy source-channel coding, a separation approach is shown in \cite{Diggavi_Shamai:2010} to be
optimal or approximately optimal to communicate {\em independent} sources. In
\cite{Cover_ElGamal_Salehi:1980}, on the other hand, a sufficient condition based on joint
source-channel coding to send correlated sources over a MAC is given, along with an uncomputable
expression for the outer bound. As the sufficient condition in \cite{Cover_ElGamal_Salehi:1980}
provides a region greater than that ensured to be achieved by separate source and channel coding,
it is proved that the separate source-channel coding is {\em not} optimal for {\em correlated}
sources. In \cite{Abdallah:2008}, \cite{FadiAbdallah_Caire:2008}, however, the authors show that
performing {\em separate} source and channel coding for the important case of a Gaussian MAC with
phase shifts, shown in Fig. \ref{PI_MAC} is optimal. Namely, in \cite{Abdallah:2008} and
\cite{FadiAbdallah_Caire:2008}, F. Abi Abdallah et. al. showed the following separation theorem for
a class of phase asynchronous multiple access channels for both non-ergodic, and ergodic i.i.d.
phase fading:

\begin{theorem}\label{Theorem_Capacity_MAC}
{\em Reliable communication over a PI-MAC}: A necessary condition for reliable communication of the
source pair $(\bU,\bV) \sim {\prod_{i}}p(u_{i},v_{i})$ over a class of multiple access channels
with unknown phase fading at the transmitters, with power constraints $P_1,P_2$ on the
transmitters, and fading amplitudes $g_1,g_2>0$, is given by
\begin{align}
H(U \vert V) & \leq \log(1+g_1^2P_1/N), \label{separation__MAC_1}\\
H(V \vert U) & \leq \log(1+g_2^2P_2/N), \\
H(U,V) & \leq \log(1+(g_1^2P_1+g_2^2P_2)/N), \label{separation_MAC_3}
\end{align}
\noindent where $N$ is the noise power. Sufficient conditions for the reliable communications are
also given by \eqref{separation__MAC_1}-\eqref{separation_MAC_3}, with $\leq$ replaced by $<$.
\thmend
\end{theorem}

{Also, the recent work \cite{Dabora:2011} addresses the same problem for a phase fading Gaussian
multiple access relay channel (MARC) and proves a separation theorem under some channel coefficient
conditions. For the achievability part, the authors use the results of
\cite{Kramer_Gupta_Kumar_famous:2005}, \cite{Kramer_MARC_Main1:2004}, and
\cite{Kramer_MARC_Main2:2004} based on a combination of {\em regular} Markov encoding at the
transmitters and {\em backward} decoding at the receiver \cite{Kramer_WRC:2003}. In particular, in
order to derive the achievable region for discrete-memoryless MARC, the authors of
\cite{Kramer_MARC_Main1:2004} use codebooks of the same size which is referred to as regular Markov
encoding. This is in contrast with block Markov encoding which was introduced by Cover and El Gamal
in \cite{Cover:1979} for the relay channel. There, the encoding is done using codebooks of
different sizes and is referred to as {\em irregular} block Markov encoding.

In this paper, we consider a more general PI-MARC, in which one of the encoders is helped by the
other one causally or non-causally. We refer to such networks as phase-incoherent unidirectional
cooperative MARCs or PI-UC-MARCs for short. Furthermore, we also prove separation theorems for a
phase-incoherent interference channel (PI-IC) under strong interference conditions and
phase-incoherent interference relay channel (PI-IRC) under specific strong interference gain conditions.} 




The networks that we consider and for which we prove our results are listed as follows:
\begin{itemize}
\item PI-UC-MARC with {\em non-causal} (NC) cooperation between transmitters and with strong
    path gains from transmitters to the relay. We refer to this network as phase incoherent
    unidirectional non-causal cooperative (PI-UNCC)-MARC. By removing the relay, the results
    can be specialized to the case of a MAC (PI-UNCC-MAC).
\item PI-UC-MARC with {\em causal} (C) cooperation between transmitters and with strong path
    gains from transmitters to the relay. This network is called a phase-incoherent
    unidirectional causal cooperative (PI-UCC)-MARC. By removing the relay, the results can be
    specialized to the case of a MAC (PI-UCC-MAC).
\item Phase incoherent interference channel (PI-IC) in strong interference regime.
\item Phase incoherent interference relay channel (PI-IRC) in a specific strong interference
    regime with strong path gains from transmitters to the relay.
\end{itemize}

We show that if the phase shifts are unknown to the transmitters, then the optimal performance is
no better than the scenario in which the information sources are first source coded and then
channel coded separately, i.e., the correlation between the sources is not helpful to enlarge the
achievable region, as opposed to cases where the transmitters have knowledge of the phase shifts
and could potentially use beamforming, for example, to joint source-channel code the data and
achieve higher rates. Although we assume non-ergodic phase shifts throughout the paper, as in
\cite{FadiAbdallah_Caire:2008}, our results are also true for the ergodic case, where the phases
change i.i.d. from symbol to symbol. The contributions of this work are stated in the form of four
separation theorems that are given in the following sections.

Further, we conjecture that optimality of separation is true not only for the specific gain
conditions we state, but also for all possible values of path gains. Hence, we conjecture that
separation is optimal for unrestricted forms of the phase incoherent Gaussian phase fading channels
discussed in this paper. The approach we used here to prove the separation theorems which is based
on computing necessary and sufficient conditions for reliable communication, however, may not be
viable to
prove the conjecture. 

The rest of this paper is organized as follows. We introduce the phase asynchronous multi-user
networks considered in this work in Section \ref{preliminaries} along with a key lemma that we use
several times in the paper. In Section \ref{Section-PI-MARC}, we define the general problem of the
joint source-channel coding for a PI-MARC and state a separation theorem for it. In Sections
\ref{Section-PI-UNCC-MARC} and \ref{Section-PI-UCC-MARC}, we state and prove separation theorems
under specific gain conditions for a class of phase asynchronous MARCs in which the encoders
cooperate unidirectionally both non-causally and causally respectively. Next, In Sections
\ref{Section-PI-IC} and \ref{Section-PI-IRC}, we consider joint source-channel coding problem for
interference channels and interference relay channels under phase uncertainty respectively and
likewise state and prove separation theorems for them under strong interference conditions. We
finally conclude the results in Section \ref{Conclusion} along with a conjecture.


\section{Network Models and a Key Lemma}\label{preliminaries}
Consider two finite alphabet sources $\{U_i,V_i\}$ with correlated outputs that are drawn according
to a distribution $P[U_{i}=u,V_{i}=v] = p(u,v)$. The sources are memoryless, i.e., $(U_i,V_i)$'s
are independent and identically distributed (i.i.d). Both of the sources are to be transmitted to
the corresponding destinations through continuous alphabet and discrete-time memoryless non-ergodic
Gaussian channel models. Channels are parameterized by the phase shifts that are introduced by
different paths of the network which are, as a realistic assumption for wireless networks, not
known to the transmitters. The vector $\Bt$ denotes the non-ergodic phase fading parameters. For
simplicity, throughout the paper, we assume that transmitter node with index $i \in \{1,2,r\}$ has
power constraint $P_i$ and the noise power at all corresponding receiving nodes is $N$.

In the models that we consider, the receiver(s) are fully aware of $\Bt$. However, the transmitters
do not have access to the channel state information (CSI), $\Bt$, but only the knowledge of the
family of channels over which the communication is done and the code design must be robust for all
$\Bt$. Such channels are referred to as {\em compound} channels \cite{Csiszar:1981},
\cite{Wolfowitz:1978}. Nevertheless, in order to avoid ambiguity, we call the particular channel
under consideration a {\em phase-incoherent} (PI) channel with correlated sources. In the sequel,
we introduce the channel models
to be considered in this paper. 
\subsection{Multiple Access Channel (MAC)}\label{Introduction_MAC}

\begin{figure}
\hspace{0cm}\resizebox*{!}{2.2in}{\centering{\includegraphics*[angle = 90,viewport=-30 -168 200
600]
{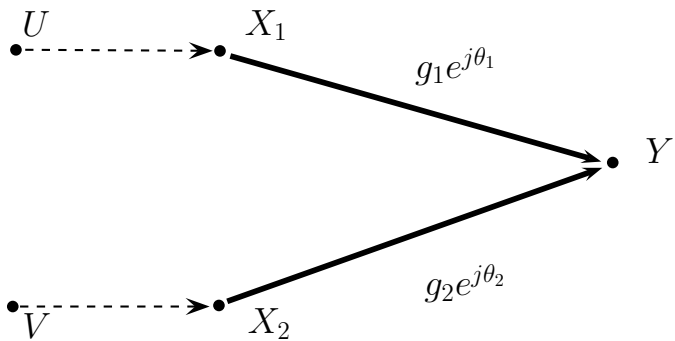}} }
\caption{Correlated sources and phase incoherent multiple access channel.} \label{PI_MAC}\vspace{0cm}
\end{figure}

A phase incoherent multiple access channel (PI-MAC)
$\left({\mathcal{X}_1}\times{\mathcal{X}_2},\mathcal{Y},p_{\bm\theta}(y\vert x_1,x_2)\right)$ with
parameter $\Bt  = (\theta_1,\theta_2) \in [0,2\pi)^{2}$ is illustrated in Fig. \ref{PI_MAC}. The
MAC is described by the relationship
\begin{align}\label{MAC channel-model}
Y_i & = h_1X_{1i} + h_2X_{2i} + Z_i,
\end{align}
\noindent where $X_{1i}, X_{2i}, Y_{i} \in \mathbb{C}$, $Z_i \sim {\mathcal{C}\mathcal{N}(0,N)}$ is
circularly symmetric complex Gaussian noise, $h_1 = g_1e^{j\theta_1},h_2=g_2e^{j\theta_2}$ are
non-ergodic complex channel gains, and parameter ${\bm \theta}$ represents the phase shifts
introduced by the channel to inputs $X_1$ and $X_2$, respectively. The amplitude gains, $g_1$ and
$g_2$, are assumed to be known at transmitters and can model e.g., line of sight path gains.

\subsection{MAC with Unidirectional Non-Causal Cooperation Between Transmitters (UNCC-MAC)}\label{Section_introduction_UNCC-MAC}

A PI-UNCC-MAC $\left({\mathcal{X}_1}\times{\mathcal{X}_2},\mathcal{Y},p_{\bm\theta}(y\vert
x_1,x_2)\right)$ is depicted in Fig. \ref{PI_UNCCMAC}. The first encoder $X_1$ has non-causal and
perfect knowledge of the second source $V$. The channel characteristic is the same as an ordinary
PI-MAC given in \eqref{MAC channel-model}.

\begin{figure}
\hspace{1cm}\resizebox*{!}{1.9in}{\centering{\includegraphics*[angle = 90,viewport=10 -168 215
600]
{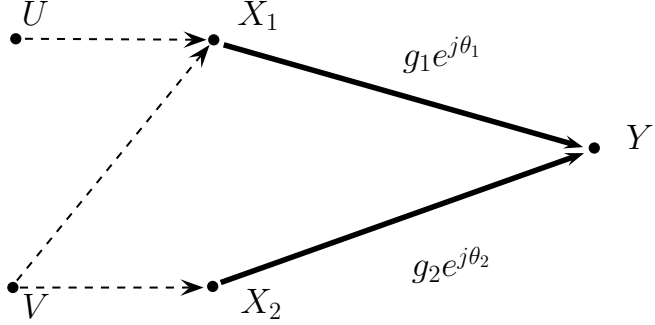}} }
\caption{Correlated sources and phase incoherent unidirectional non-causal cooperative multiple access channel.} \label{PI_UNCCMAC}\vspace{0cm}
\end{figure}

\subsection{MAC with Unidirectional Causal Cooperation Between Transmitters (UCC-MAC)}\label{Section_introduction_UCC-MAC}

Another multi-user model that is considered in this paper is a cooperative variation of the
multiple access channel
$\left({\mathcal{X}_1}\times{\mathcal{X}_2},\mathcal{Y}_{1}\times\mathcal{Y},p_{\bm\theta}(y_{1},y\vert
x_1,x_2)\right)$ where one of the transmitters can play the role of {\bf \underline a} {\em relay}
for the other. This channel model is shown in Fig. \ref{PI_UCC-MAC} where the first transmitter
(node indicated by $X_1$) can help the second transmitter (node indicated by $X_2$) to transmit its
information to the destination. However, node $2$ cannot help node $1$ and thus we refer to such a
channel as a unidirectional cooperative MAC. The received signal of the PI-UCC-MAC at the
destination is also given by \eqref{MAC channel-model}. At the transmitter/relay node, node $1$, we
have
\begin{align}
Y_{1i} = g_{21}e^{j\theta_{21}} X_{2i} + Z_{1i},
\end{align}
\noindent {where $g_{21}$ and $\theta_{21}$ are the path gain and the phase shift of the channel
from node $2$ to node $1$ respectively. The vector $\Bt$ for the PI-MAC has three elements and is
defined as $\Bt = (\theta_{1},\theta_{2},\theta_{21})$.}

\begin{figure}
\vspace{0cm} \hspace{1cm}\resizebox*{!}{2.2in}{\centering{\includegraphics*[angle = 90,viewport=-20
-168 198 600]
{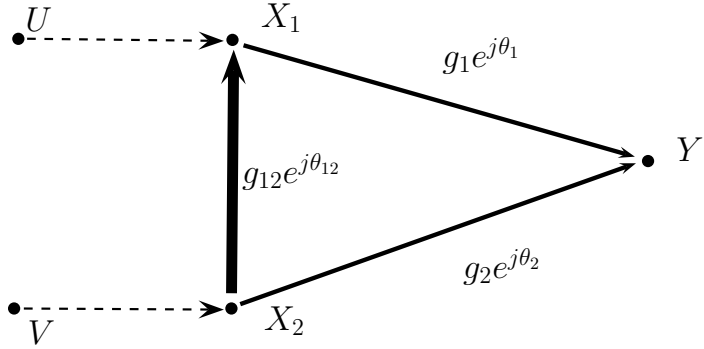}} }
\vspace{-1cm}\caption{Correlated sources and phase incoherent unidirectional causal cooperative
multiple access channel.} \label{PI_UCC-MAC}\vspace{0cm}
\end{figure}

\subsection{Multiple Access Relay Channel (MARC)}\label{Section_introduction_MARC}

A multiple access relay channel is a network with four nodes, two transmitters, a relay and a
destination. As depicted in Fig. \ref{PI_MARC}, in a MARC with phase fading
$\left({\mathcal{X}_1}\times{\mathcal{X}_2}\times{\mathcal{X}_r},\mathcal{Y},p_{\bm\theta}(y,y_{r}\vert
x_1,x_2,x_r)\right)$, two transmitters wish to reliably send their information to a common
destination, with the help of a relay. There are five paths in the network. The phase parameters
are not known to the transmitters and hence we refer to the MARC as PI-MARC. The received signal at
the destination is given by

\begin{align}\label{MARC-channel-model}
Y_i & = h_1X_{1i} + h_2X_{2i} + h_rX_r + Z_i,
\end{align}
\noindent where $X_{1i}, X_{2i}, Y_{i} \in \mathbb{C}$, $Z_i \sim {\mathcal{C}\mathcal{N}(0,N)}$ is
circularly symmetric complex Gaussian noise, $h_1 =
g_1e^{j\theta_1},h_2=g_2e^{j\theta_2},h_r=g_re^{j\theta_r}$ are non-ergodic complex channel gains,
and $\theta_1,\theta_2,\theta_r$ represent the phase shifts introduced by the channel to inputs
$X_1$, $X_2$ and $X_r$, respectively. 

Moreover, the signal received at the relay can be written as
\begin{align}\label{MARC-channel-model-relay}
Y_{ri} & = h_{1r}X_{1i} + h_{2r}X_{2i} + Z_{ir}
\end{align}
\noindent where $Z_{ir} \sim {\mathcal{C}\mathcal{N}(0,N)}$ and $h_{1r} =
g_{1r}e^{j\theta_{1r}},h_{2r}=g_{2r}e^{j\theta_{2r}}$ are the complex path gains with unknown
phases $\theta_{1r},\theta_{2r}$ at transmitters. The parameter $\Bt =
(\theta_1,\theta_2,\theta_r,\theta_{1r},\theta_{2r}) \in [0,2\pi)^{5}$ of the PI-MARC includes all
of the fading phases in different paths.

\begin{figure}
\hspace{1cm}\resizebox*{!}{2.2in}{\centering{\includegraphics*[angle = 90,viewport=-20 -158 198
600]
{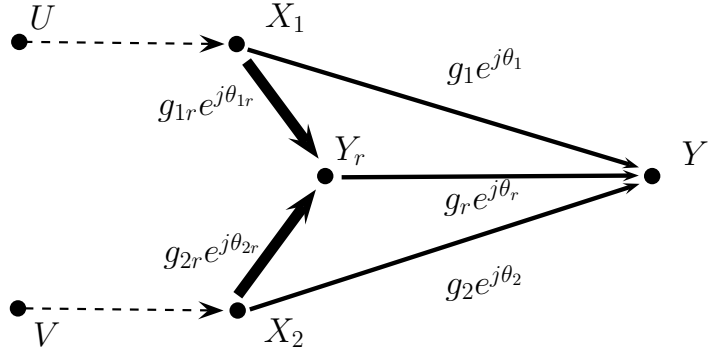}} }
\caption{Correlated sources and phase incoherent multiple access relay channel.} \label{PI_MARC}
\end{figure}

\subsection{MARC with Unidirectional Non-Causal Cooperation Between Transmitters (UNCC-MARC)}\label{Section_introduction_UNCC-MARC}

An important multi-user network that we consider is a unidirectional cooperative MARC, in which the
first encoder has non-causal access to the second source $V$. Indeed, UNCC-MARC is a UNCC-MAC with
a relay. The channel model is similar to an ordinary MARC, but the setup of the sources and
encoders are different. Fig. \ref{PI_UNCCMARC} depicts a PI-UNCC-MARC.

Like PI-MARC, the input/output relationships of the channel for the receiver and the relay are
given by \eqref{MARC-channel-model} and \eqref{MARC-channel-model-relay}. The parameter $\Bt$ is
the same as that of the ordinary MARC.

\begin{figure}
\hspace{1cm}\resizebox*{!}{2.2in}{\centering{\includegraphics*[angle = 90,viewport=-20 -158 198
600]
{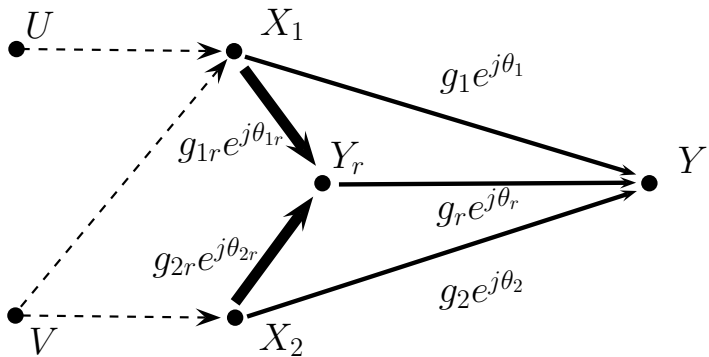}} }
\caption{Correlated sources and phase incoherent multiple access relay channel with unidirectional non-causal
cooperation between the encoders.} \label{PI_UNCCMARC}
\end{figure}

\subsection{MARC with Unidirectional Causal Cooperation Between Transmitters (UCC-MARC)}\label{Section_introduction_UNCC-MARC}

We also consider sending sources $U,V$ over a PI-MARC with causal unidirectional cooperation
between the encoders denoted by
$\left({\mathcal{X}_1}\times{\mathcal{X}_2}\times{\mathcal{X}_r},\mathcal{Y}_{1}\times\mathcal{Y}_{r}\times\mathcal{Y},p_{\bm\theta}(y_{1},y_{r},y\vert
x_1,x_2,x_r)\right)$. As it can be seen from Figure \ref{PI_UCCMARC}, the encoder $X_1$ does not
have non-causal knowledge about $V$, but it receives a noisy phase faded version of $X_2$ through
the link from node $2$ to node $1$. Again, \eqref{MARC-channel-model} and
\eqref{MARC-channel-model-relay} describe the input/output relationships of the channel for the
receiver and the relay. Additionally, the relationship
\begin{align}
Y_{1i} = g_{21}e^{j\theta_{21}} X_{2i} + Z_{1i}
\end{align}
\noindent describes the cooperative link from node $2$ to node $1$ which completes the definition
of a PI-UCC-MARC. The parameter $\Bt$ for the PI-UCC-MARC is the vector
$\Bt=(\theta_1,\theta_2,\theta_r,\theta_{1r},\theta_{2r},\theta_{12})$.

\begin{figure}
\hspace{1cm}\resizebox*{!}{2.2in}{\centering{\includegraphics*[angle = 90,viewport=-20 -158 198
600]
{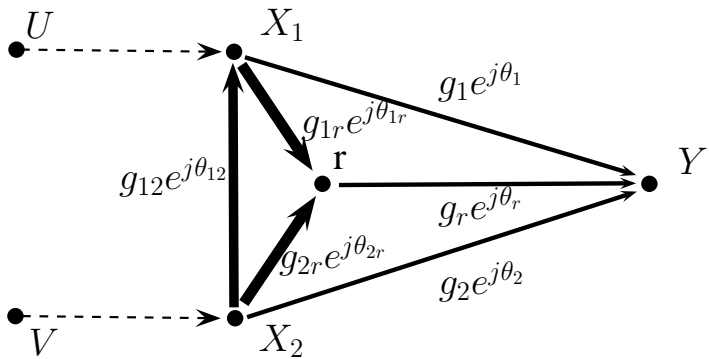}} }
\caption{Correlated sources and phase incoherent multiple access relay channel with unidirectional causal
cooperation between the encoders.} \label{PI_UCCMARC}
\end{figure}

\subsection{Interference Channel (IC)}\label{Section_introduction_IC}

Another network model we consider in this paper is the two-user interference channel with strong
interference. A continuous alphabet, discrete-time memoryless interference channel (IC) with phase
fading is denoted by
$\left({\mathcal{X}_1}\times{\mathcal{X}_2},\mathcal{Y}_{1}\times\mathcal{Y}_{2},p_{\Bt_1,\Bt_2}(y_{1},y_{2}\vert
x_1,x_2)\right)$ and its probabilistic characterization is described by the relationship
\begin{align}\label{IC-channel-model1}
Y_{1i} & = g_{11}e^{j\theta_{11}}X_{1i} + g_{21}e^{j\theta_{21}}X_{2i} + Z_{1i}, \\
Y_{2i} & = g_{12}e^{j\theta_{12}}X_{1i} + g_{22}e^{j\theta_{22}}X_{2i} + Z_{2i},\label{IC-channel-model2}
\end{align}
\noindent where $X_{1i}, X_{2i}, Y_{i} \in \mathbb{C}$, $Z_i \sim {\mathcal{C}\mathcal{N}(0,N)}$ is
circularly symmetric complex Gaussian noise, $g_{11},g_{12},g_{21},g_{22}$ are non-ergodic complex
channel gains, and parameters ${\Bt_1} = (\theta_{11},\theta_{21}) \in [0,2\pi)^{2}$, ${\Bt_2} =
(\theta_{12},\theta_{22}) \in [0,2\pi)^{2}$ represents the phase shifts introduced by the channel
to inputs $X_1$ and $X_2$, respectively. Figure \ref{PI_IC} depicts such a channel. We refer to the
IC defined by \eqref{IC-channel-model1} and \eqref{IC-channel-model2} as PI-IC if we assume the
phase shift parameters $\Bt_1,\Bt_2$ are not known to the transmitters.

\begin{figure}
\hspace{1cm}\resizebox*{!}{2.2in}{\centering{\includegraphics*[angle = 90,viewport=-20 -168 198
600]
{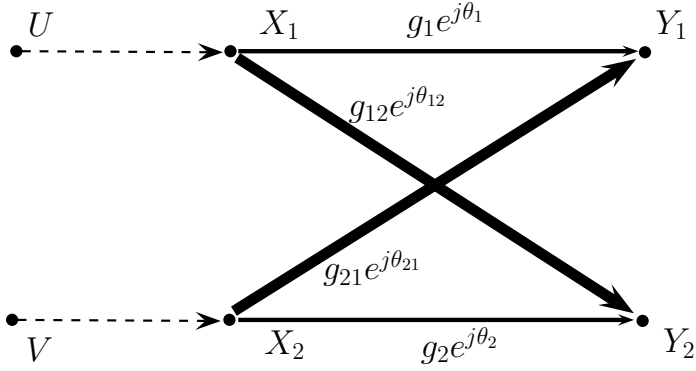}} }
\caption{Correlated sources and phase incoherent interference channel} \label{PI_IC}\vspace{0cm}
\end{figure}

\subsection{Interference Relay Channel (IRC)}\label{Section_introduction_IRC}

The last network model we consider is an interference channel with two transmitters and a relay
referred to as {\em interference relay channel} (IRC), depicted in Figure \ref{PI_IRC}. Again, we
consider phase fading at all paths, unknown to the transmitters and thus call the channel a
phase-incoherent IRC (PI-IRC). The PI-IRC $({\mathcal{X}_1}\times {\mathcal{X}_2} \times
{\mathcal{X}_r} , {\mathcal{Y}_1}\times{\mathcal{Y}_2}\times{\mathcal{Y}_r},
p_{\bm\theta}({y_1,y_2,y_r}\vert x_1,x_2,x_r))$ is described by relationships
\begin{align}
Y_{1i} & = g_{11}e^{j\theta_{11}}X_{1i} + g_{21}e^{j\theta_{21}}X_{2i} + g_{r1}e^{j\theta_{r1}}X_{ri} + Z_{1i},\nonumber\\
Y_{2i} & = g_{12}e^{j\theta_{12}}X_{1i} + g_{22}e^{j\theta_{22}}X_{2i} + g_{r2}e^{j\theta_{r2}}X_{ri} + Z_{2i},\nonumber\\
Y_{ri} & = g_{1r}e^{j\theta_{1r}}X_{1i} + g_{2r}e^{j\theta_{2r}}X_{2i} + Z_{ri},\nonumber
\end{align}
\noindent where $X_{1i}, X_{2i}, X_{ri}, Y_{1i}, Y_{2i}, Y_{ri} \in \mathbb{C}$, $Z_{1i}, Z_{2i},
Z_{ri} \sim {\mathcal{C}\mathcal{N}(0,N)}$ are circularly symmetric complex Gaussian noises,
$g_{11},$ $g_{21},$ $g_{r1},$ $g_{12},$ $g_{22},$ $g_{r2}$ are non-ergodic complex channel gains,
and parameter ${\Bt} = (\theta_{11},\theta_{21},\theta_{r1},\\ \theta_{12},\theta_{22},
\theta_{r2},\theta_{1r},\theta_{2r}) \in [0,2\pi)^{8}$ represents the phase shifts introduced by
the channel to inputs $X_1$, $X_2$ and $X_r$, respectively.

\begin{figure}
\hspace{1cm}\resizebox*{!}{2.2in}{\centering{\includegraphics*[angle = 90,viewport=-20 -218 195
600]
{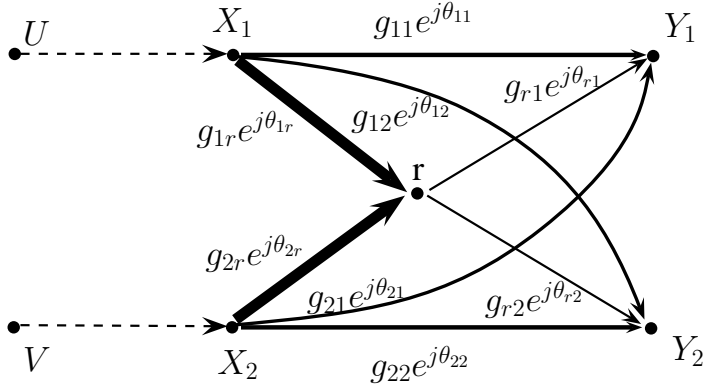}} }
\caption{Correlated sources and phase incoherent interference relay channel} \label{PI_IRC}\vspace{0cm}
\end{figure}

\subsection{Key Lemma}\label{key_lemma_section}

\begin{definition}
Let $\bX{} = (X_1,X_2,\cdots,X_m)$, be a vector of random variables with joint distribution
$p_{\bX{}}$ and $\max_{i} {\mathbb{E}}\|X_{i}\|^{2} \leq \infty$. 
Also let the scalar RV $V \triangleq \sum_{i=1}^{m}g_{i}e^{j\theta_{i}}X_{i} + Z$, where
$g_ie^{j\theta_{i}}$ are arbitrary complex coefficients and $Z\sim \mathcal{CN}(0,N)$. \thmend
\end{definition}

We now state the following lemma which asserts that the minimum over $\Bt =
(\theta_1,\theta_2,\cdots,\theta_m)$ of the mutual information between $\bX{}$ and $V$, is
maximized when $\bX{}$ is a zero-mean Gaussian vector with independent elements, i.e., RVs
$X_1,X_2,\cdots,X_m$ are independent Gaussians with zero mean.

{\em Notation}: For convenience, we denote the mutual information between $\bX{}$ and $V$  by
\begin{align}
B_{\bm\theta}(p_{\bX{}}) \triangleq I(\bX{};\sum_{i=1}^{m}g_{i}e^{j\theta_{i}}X_{i} + Z).\nonumber
\end{align}


\begin{lemma}\label{Lemma_Gaussian_Optimal}
Let ${\mathcal{P}} = \{p_{\bX{}}: \mathbb{E}\|X_{i}\|^{2} \leq P_{i}, \forall i\}$ and $p^*_{\bX{}}
\in {\mathcal{P}}$ be a zero-mean Gaussian distribution with independent elements and
$\mathbb{E}\|X_{i}\|^{2} = P_{i}, \forall i$. Then,
\begin{align}
\max_{p_{\bX{}} \in {\mathcal{P}}} \min_{\Bt} B_{\bm\theta}(p_{\bX{}}) = \log\left({1 +}\sum_{i=1}^{m} g_i^2{P_i/N}\right)
 = B_{\bm\theta}(p^{*}_{\bX{}}), \nonumber
\end{align}
\noindent i.e., when $\Bt$ is chosen adversarially, the best ${\bX{}}$ is a zero-mean Gaussian
vector with independent elements and ${\rm Var}(X_i) = P_i, \forall i$. \thmend
\end{lemma}
\begin{proof}

By definition, we have
\begin{align}
B_{\bm\theta}(p_{\bX{}}) & = h(g_1e^{j{\theta_1}}X_{1}+g_2e^{j{\theta_2}}X_{2}+\cdots+g_Ne^{j{\theta_m}}X_{m}+Z) - h(Z). \nonumber
\end{align}
\noindent By letting $\mathbb{E}(X_{i}X_{j}) = \rho_{ij}\sqrt{P_{i}P_{j}}$, it can be easily seen
that the RV $V$ has a fixed variance $\sigma^{2}_{V}$ which is equal to

\begin{align}\label{V_variance}
{\sigma^2_{V}} = \left(\sum_{i=1}^{m} g_i^2{P_i}\right)+N+2\sum_{i<j} g_ig_j{{\sqrt{{P_i}{P_j}}}} \ \Re {\left\{\rho_{ij}
e^{j(\theta_i-\theta_j)}\right\}}.
\end{align}
\noindent Using the fact that for a given variance $\sigma^2_V$, the Gaussian distribution
maximizes the differential entropy $h(V)$ \cite{Cover:2006}, we can bound
$B_{\bm\theta}(p_{\bX{}})$ as

\begin{align}\label{Gaussian_more_than_anything_equation}
B_{\bm\theta}(p_{\bX{}}) & \leq {1 \over 2} \log (2\pi e \sigma^2_{V}) - h(Z).
\end{align}

Next, note that $\min_{\Bt} \sigma^{2}_{V}$ is maximized when $\rho_{ij} = 0, \forall i,j$. It can
be seen from \eqref{V_variance} that if $\rho_{ij} \neq 0$, the parameters
$\theta_{1},\theta_{2},\cdots,\theta_{m}$ can be chosen such that the term $2\sum_{i<j}
g_ig_j{{\sqrt{{P_i}{P_j}}}} \ \Re {\left\{\rho_{ij} e^{j(\theta_i-\theta_j)}\right\}}$ is strictly
negative. Therefore, independent Gaussians ($\rho_{ij} = 0, \forall i,j$) maximize the right hand
side of \eqref{Gaussian_more_than_anything_equation} and the lemma is proved.
\end{proof}

\begin{remark}\label{remark_key_lemma_ergodic}
For the ergodic setting, where $\Bt$ is i.i.d. from channel use to channel use, uniformly
distributed over $[0,2\pi)^{m}$, and the {\em averaged} mutual information over $\Bt$ is to be
maximized, a similar result is given in \cite[Thm. 2]{Kramer_Gupta_Kumar_famous:2005}.
Specifically,
\begin{align}
\max_{p_{\bX{}}} {\mathbb{E}}_{\Bt} B_{\bm\theta}(p_{\bX{}}) = \log\left({1 +}\sum_{i=1}^{m} g_i^2{P_i/N}\right). \nonumber
\end{align}
\thmend
\end{remark}

\section{Phase Incoherent Multiple access Relay Channel} \label{Section-PI-MARC}

In this section, we formulate the problem of source-channel coding for the PI-MARC introduced in
Section \ref{Section_introduction_MARC} and state a separation theorem for it \cite{Dabora:2011}.
The definitions and problem formulation given in this section will be of use for the other networks
in the paper.

\subsection{Preliminaries} \label{definitions}

\begin{definition}
{\em Joint source-channel code}: A joint source-channel code of length $n$ for the PI-MARC
introduced in Section \ref{Section_introduction_MARC} with correlated sources is defined by
\begin{enumerate}
\item {Two encoding functions
\begin{align*}
\left(x_{11},x_{12},\cdots,x_{1n}\right) = {\bbx^{n}_{1}}:& \mathcal{U}^n \rightarrow \mathcal{X}_{1}^n\\
\left(x_{21},x_{22},\cdots,x_{2n}\right) = {\bbx^{n}_{2}}:& \mathcal{V}^n \rightarrow \mathcal{X}_{2}^n,
\end{align*}
\noindent that map the source outputs to the codewords. Furthermore, we define relay encoding
functions by
\begin{align}
x_{ri} = f_{i}(y_{r1},y_{r2},\cdots,y_{r(i-1)}), \ \ i=1,2,\cdots,n. \nonumber
\end{align}

The sets of codewords are denoted by the codebook $\mathcal{C} = \{(\bx{1}(\bu),\bx{2}(\bv)) :
\bu \in \mathcal{U}^n, \bv \in \mathcal{V}^n\}$}.
\item{Power constraint $P_1$, $P_2$ and $P_r$ at the transmitters, i.e.,
\begin{align}\label{Power_constraint}
\mathbb{E}\left[{1 \over n} \sum_{i=1}^{n}\|X_{ji}\|^2\right] \leq P_j, \ j=1,2,r,
\end{align}
\noindent where $\mathbb{E}$ is the expectation operation over the distribution induced by $\bU^{n}, \bV^{n}$}. 
\item{A decoding function
\[g^n_{\bm\theta} : \mathcal{Y}^n \rightarrow  \mathcal{U}^n \times  \mathcal{V}^n. \]
}
\end{enumerate}
\thmend
\end{definition}
{ Upon reception of the received vector $\bY^n$, the receiver decodes $(\hat{\bU}^n,\hat{\bV}^n) =
g_{\bm\theta}(\bY^n)$ as the transmitted source outputs. The probability of an erroneous decoding
depends on $\Bt$ and is given by}
\begin{align}
P_e^n(\bm\theta) & = P\{(\bU^n,\bV^n) \neq (\hat{\bU}^n,\hat{\bV}^n) \vert {\bm\theta}\} \nonumber\\
& = \sum_{({\bu}^n,{\bv}^n) \in  {\mathcal{U}^n \times   \mathcal{V}^n}}p({\bu}^n,{\bv}^n) \times
P\{(\hat{\bU}^n,\hat{\bV}^n) \neq (\bu^n,\bv^n) \ \vert ({\bu}^n,{\bv}^n) , {\bm\theta}\}. \nonumber
\end{align}


\begin{definition} We say the source $\{U_i,V_i\}_{i=1}^{n}$ of i.i.d. discrete random variables
with joint probability mass function $p(u,v)$ {\em can be reliably sent} over the PI-MARC, if there
exists a sequence of encoding functions
${\mathcal{E}_{n}}\triangleq\{\bbx^{n}_{1}(\bU^{n}),\bbx^{n}_{2}(\bV^{n}),f_{1},f_{2},\cdots,f_{n}\}$
and decoders $g^n_{\bm\theta}$ 
such that the output sequences $\bU^{n}$ and $\bV^{n}$ of the source can be estimated with
asymptotically small probability of error ({uniformly} over all parameters $\bm\theta$) at the
receiver side from the received sequence $\bY^{n}$, i.e.,
\begin{align}\label{main_error_probability}
\left[\sup_{\bm\theta} P_e^n(\bm\theta)\right] \longrightarrow 0, \ \ {\rm as}   \  \ n \rightarrow \infty.
\end{align}
\thmend
\end{definition}

\begin{theorem}\label{Theorem_Capacity_MARC}{\em Reliable communication over a PI-MARC}: Consider a PI-MARC with power constraints $P_1,P_2,P_r$ on
the transmitters, fading amplitudes $g_1,g_2,g_r>0$ between the nodes and the receiver and
$g_{1r},g_{2r}>0$ between the transmitter and the relay, and the gain conditions
\begin{align}\label{Equation1_Condition_MARC}
g^{2}_{1r}P_1 & \geq g^{2}_{1}P_1 + g^{2}_{r}P_r, \\
g^{2}_{2r}P_1 & \geq g^{2}_{2}P_1 + g^{2}_{r}P_r.\label{Equation2_Condition_MARC}
\end{align}

A necessary condition for reliably sending the source pair $(\bU,\bV) \sim
{\prod_{i}}p(u_{i},v_{i})$, over a PI-MARC, is given by
\begin{align}
H(U \vert V) & \leq \log(1+(g_1^2P_1+g_r^2P_r)/N), \label{separation_1_MARC}\\
H(V \vert U) & \leq \log(1+(g_2^2P_2+g_r^2P_r)/N), \\
H(U,V) & \leq \log(1+(g_1^2P_1+g_2^2P_2+g_r^2P_r)/N). \label{separation_3_MARC}
\end{align}
\noindent Moreover, \eqref{separation_1_MARC}-\eqref{separation_3_MARC} also describes sufficient
conditions for reliable communications with $\leq$ replaced by $<$. \thmend
\end{theorem}

\begin{proof}
The theorem is the same as \cite[Theorem 4]{Dabora:2011}.
\end{proof}


In Sections \ref{Section-PI-UNCC-MARC} and \ref{Section-PI-UCC-MARC}, we study a more general
version of the PI-MARC in which the transmitters cooperate in a specific way. Indeed, we consider a
pair of correlated sources to be communicated over a phase incoherent (PI) multiple access relay
channel where one of the transmitters has {\em causal} or {\em non-causal} side information about
the message of the other. We refer to such channels as UC-MARC. In the non-causal case (see Fig.
\ref{PI_UNCCMAC}), there is no path between the transmitters and the first encoder knows both
sources outputs $\bU,\bV$, whereas in the causal case (see Fig. \ref{PI_UCCMARC}), the first
transmitter works as a relay for the other while communicating its own information. For the
situations where the channel gains between the relay and the transmitters are large enough, we
prove that the separation approach is optimal. This may correspond to the physical proximity of the
relay and the transmitters to each other. For the causal case, we have an additional condition on
the gain between the encoders. The phase fading information is not known to the transmitters while
it is known at the receivers.

\section{UC-MARC with non-causal side information} \label{Section-PI-UNCC-MARC}

In this section, we study the PI-UNCC-MARC introduced in Section
\ref{Section_introduction_UNCC-MARC} with the pair of arbitrarily correlated sources $(U,V)$. The
definition of a joint source-channel code for the PI-UNCC-MARC is identical to the one defined for
a PI-MARC in section \ref{definitions} except for the definition of the encoding function $\bx{1}$
which is replaced by

\begin{align}
\left(x_{11},x_{12},\cdots,x_{1n}\right) = \bbx^{n}_{1}(\bU,\bV). \nonumber
\end{align}


%

\begin{theorem}\label{Theorem_Capacity_MARC_noncausal}{\em Reliable Communication over a
PI-UNCC-MARC}: Consider a PI-UNCC-MARC with non-causal cooperation and with power constraints
$P_1,P_2,P_r$ on transmitters and relay, fading amplitudes $g_1,g_2,g_r>0$ between the nodes and
the receiver and $g_{1r},g_{2r}>0$ between the transmitter and the relay. Moreover, assume the gain
conditions
\begin{align}\label{Equation1_Condition_UNCCMARC}
g^{2}_{1r}P_1 & \geq g^{2}_{1}P_1 + g^{2}_{r}P_r, \\
g^{2}_{1r}P_1 + g^{2}_{2r}P_2 & \geq g^{2}_{1}P_1 + g^{2}_{2}P_2 + g^{2}_{r}P_r. \label{Equation2_Condition_UNCCMARC}
\end{align}

A necessary condition for sending a source pair $(\bU,\bV) \sim {\prod_{i}}p(u_{i},v_{i})$, over
such PI-UC-MARC is given by
\begin{align}
H(U \vert V) & \leq \log(1+(g_1^2P_1+g_r^2P_r)/N), \label{separation_1_noncausal_MARC}\\
H(U,V) & \leq \log(1+(g_1^2P_1+g_2^2P_2+g_r^2P_r)/N). \label{separation_3_noncausal_MARC}
\end{align}
\noindent Furthermore, eqs. \eqref{separation_1_noncausal_MARC}-\eqref{separation_3_noncausal_MARC}
also give the sufficient conditions for reliable communications over such PI-UD-MARC with $\leq$
replaced by $<$. \thmend
\end{theorem}

The proof of the theorem is divided into two parts: achievability and converse. The achievability
part is obtained by a separate source and channel coding approach. The source coding part involves
Slepian-Wolf coding followed by a channel coding technique which is based on the block Markov
coding. 
The converse and achievability parts of Theorem \ref{Theorem_Capacity_MARC} are discussed and
proved in the sequel.

\subsection{Converse}\label{Section_Converse_MARC}

We derive an outer bound on the capacity region of the PI-UC-MARC (both causal and non-causal)
under gain conditions \eqref{Equation1_Condition_UNCCMARC}-\eqref{Equation2_Condition_UNCCMARC} and
prove the converse part of Theorem \ref{Theorem_Capacity_MARC_noncausal}.

\begin{lemma}\label{Lemma_Converse_MARC}
{\em Converse}: Let $\mathcal{E}_{n}$, and $g^n_{\bm\theta}$ be a sequence in $n$ of encoders and
decoders for the PI-UC-MARC for which $ \sup_{\bm\theta} P_e^n(\bm\theta) \longrightarrow 0$, as $n
\rightarrow \infty$. Then
\begin{align}
H(U\vert V) &\leq \min_{\Bt} I(X_1,X_r;g_1e^{j\theta_1}X_1+g_re^{j\theta_r}X_r+Z), \nonumber \\
H(U,V) &\leq \min_{\Bt} I(X_1,X_2,X_r;g_1e^{j\theta_1}X_1+ g_2e^{j\theta_2}X_2+g_re^{j\theta_r}X_r+Z), \nonumber
\end{align}
\noindent for some {\em joint} distribution $p_{X_1,X_2,X_r}$ such that $\mathbb{E}\vert X_1\vert^2
\leq P_1, \mathbb{E}\vert X_2\vert^2 \leq P_2, \mathbb{E}\vert X_r\vert^2 \leq P_r$, with $Z \sim
\mathcal{C}\mathcal{N}(0,N)$. \thmend
\end{lemma}

\begin{proof}

First, fix a PI-UC-MARC with given parameter $\Bt$, a codebook $\mathcal{C}$, and induced {\em
empirical} distribution $p_{\bm\theta}(\bu,\bv,\bx{1},\bx{2},\bx{r},\by)$ by the codebook. Since
for this fixed choice of ${\bm\theta}$, $P^n_e(\bm\theta) \rightarrow 0$, from Fano's inequality,
we have
\begin{align}\label{main_Fano_MARC}
{1 \over n}H(\bU,\bV \vert \bY,\Bt) \leq {1 \over n}{P_e^n(\bm\theta)} \log \|\mathcal{U}^{n}\times\mathcal{V}^{n}\| + {1 \over n}
\triangleq \epsilon_n(\Bt),
\end{align}
and $\epsilon_n(\Bt) \rightarrow 0$, where convergence is uniform in $\bm{\theta}$ by
\eqref{main_error_probability}. Defining $\sup_{\Bt}\epsilon_{n}(\Bt) = \epsilon_{n}$ and following
the similar steps as in \cite[Section 4]{Cover_ElGamal_Salehi:1980}, we have
\begin{align}
H(U\vert V) & = {1 \over n}H(\bU \vert \bV) \nonumber \\
& \stackrel{(\rm a)}{=} {1 \over n}H(\bU \vert \bV ,\bX{2},\Bt) \nonumber \\
& = {1 \over n}I(\bU;\bY \vert \bV, \bX{2}, \Bt) + {1 \over n}H(\bU\vert \bV,\bY,\bX{2},\Bt) \nonumber \\
& \stackrel{(\rm b)}{\leq} {1 \over n}I(\bU ; \bY \vert \bV, \bX{2}, \Bt) + \epsilon_{n} \nonumber\\
&\stackrel{(\rm c)}{\leq} {1 \over n}I(\bX{1} ; \bY \vert \bV, \bX{2}, \Bt) + \epsilon_{n}\nonumber\\
&{\leq} {1 \over n}I(\bX{1},\bX{r} ; {{\bY}} \vert \bV, \bX{2}, \Bt) + \epsilon_{n}, \label{Equation_1st_outer_main_MARC}
\end{align}
\noindent where $(\rm a)$ follows from the fact that $\bX{2}$ is only a function of $\bV$, $(\rm
b)$ follows from \eqref{main_Fano_MARC}, and $(\rm c)$ follows from data processing inequality.
Similarly, it can be shown that
\begin{align}
H(U,V) 
& = {1 \over n} I(\bU,\bV;\bY \vert \Bt) + {1\over n}H(\bU,\bV \vert \bY, \Bt) \nonumber \\
& {\leq} {1 \over n} I(\bX{1} , \bX{2}, \bX{r}; \bY \vert \Bt) + \epsilon_{n}.\label{Equation_3rd_outer_main_MARC}
\end{align}

We now define the region $C_n(\Bt)$ as
\begin{align}
C_n(\bm\theta) = \left\{(R_1, R_2): R_1 < {1 \over n} \right. I(\bX{1}^{n};\bY^{n}& \vert \bV^{n}, \bX{2}^{n}, \Bt) + \epsilon_{n},  \nonumber \\
 \qquad \qquad R_2 <   {1 \over n} \ I(\bX{1}^{n},\bX{2}^{n};&\bY^{n} \vert \Bt) + \epsilon_{n} \biggl. \biggr\}, \nonumber
\end{align}
\noindent for the empirical distribution induced by the $n$th codebook
\begin{align}
&  \prod_{i=1}^{n} p(u_i,v_i)p(\bbx_{1}^n \vert \bu)p(\bbx_{2}^n \vert \bv)
\prod_{i=1}^{n}p_{\bm\theta}(y_i,y_{ri} \vert x_{1i},x_{2i},x_{ri}) \times p(x_{ri}\vert y_{r1},y_{r2},\cdots,y_{r(i-1)}). \nonumber
\end{align}

Hence, the outer bounds \eqref{Equation_1st_outer_main_MARC} and
\eqref{Equation_3rd_outer_main_MARC} can be equivalently described by $C_n(\Bt)$:
\begin{align}
(H(U \vert V),H(U,V)) \in C_n(\Bt). \nonumber
\end{align}

We then note that the outer bound is true for all $\Bt$ and thus can be tightened by taking
intersection over $\Bt$ and letting $n \rightarrow \infty$. 
We now further upper bound $C_{n}(\Bt)$ and then take the limit and intersection.

First, we expand $\bY$ in the right hand side of \eqref{Equation_1st_outer_main_MARC} to upper
bound $H(U \vert V)$ as follows:
\begin{align}
H(U\vert V) & \ \leq {1 \over n} \ I(\bX{1},\bX{r} ; \bY \vert \bV, \bX{2}, \Bt) + \epsilon_{n} \nonumber\\
& ={1 \over n} \ I(\bX{1},\bX{r} ; g_1e^{j\theta_1}\bX{1} + g_2e^{j\theta_2}\bX{2} +
g_re^{j\theta_r}\bX{r} + \bZ \vert \bV, \bX{2})
+ \epsilon_{n} \nonumber\\
& = {1 \over n} \ I(\bX{1},\bX{r} ; g_1e^{j\theta_1}\bX{1} + g_re^{j\theta_r}\bX{r} + \bZ \vert
\bV, \bX{2})
+ \epsilon_{n} \nonumber \\
& = {1 \over n} \left[ h(g_1e^{j\theta_1}\bX{1} + g_re^{j\theta_r}\bX{r} + \bZ \vert \bV,\bX{2}) -  h(\bZ) \right] + \epsilon_{n}\nonumber \\
& \leq {1 \over n} \left[ h(g_1e^{j\theta_1}\bX{1} + g_re^{j\theta_r}\bX{r} + \bZ) - h(\bZ) \right] + \epsilon_{n} \nonumber \\
& = {1 \over n} \ I(\bX{1},\bX{r} ; g_1e^{j\theta_1}\bX{1} + g_re^{j\theta_r}\bX{r} + \bZ)
+ \epsilon_{n} \nonumber \\
& \leq {1 \over n} \sum_{i=1}^{n} I(X_{1i},X_{ri} ; g_1e^{j\theta_1}X_{1i} + g_re^{j\theta_r}X_{ri}+Z_{i})
+ \epsilon_{n} \nonumber \\
& \stackrel{(\rm a)}{=} \ I(X_{1},X_{r} ; g_1e^{j\theta_1}X_{1} + g_re^{j\theta_r}X_{r}+Z \vert W)+ \epsilon_{n} \nonumber\\
& {=} \left[ h(g_1e^{j\theta_1}X_{1} + g_re^{j\theta_r}X_{r}+Z \vert W) - h(Z) \right]+ \epsilon_{n} \nonumber\\
& {\leq} \left[ h(g_1e^{j\theta_1}X_{1} + g_re^{j\theta_r}X_{r}+Z) - h(Z) \right]+ \epsilon_{n} \nonumber\\
& = \ I(X_{1},X_{r} ; g_1e^{j\theta_1}X_{1} + g_re^{j\theta_r}X_{r}+Z)+ \epsilon_{n}, \label{HU_MARC_last}
\end{align}
\noindent where $(\rm a)$ follows by defining new random variables
\begin{align}\label{equation_new_RV_1r}
X_{j} & = X_{jW}, \ j \in \{1,2,r\},\\
Z & = Z_{W}, \\
W & \sim {\rm Uniform}\{1,2,\cdots,n\}. \label{equation_new_RV_1r_last}
\end{align}
\noindent From \eqref{Power_constraint}, the input signals $X_1, X_r$ satisfy the power constraints
\begin{align}\label{single_power_constraint}
\mathbb{E}\vert X_{j} \vert ^{2} = \mathbb{E}\left[{1 \over n} \sum_{i=1}^{n}\|X_{ji}\|^2\right] \leq P_j, \  j = 1,r,
\end{align}
\noindent and $Z \sim \mathcal{C}\mathcal{N}(0,N)$.


Moreover, following similar steps, we have
\begin{align}
H(U,V) & = {1\over n} H(\bU,\bV) \nonumber \\
& = {1 \over n} I(\bU,\bV;\bY \vert \Bt) + {1\over n}H(\bU,\bV \vert \bY, \Bt) \nonumber \\
& \leq {1 \over n} I(\bU , \bV; \bY \vert \Bt) + \epsilon_{n} \nonumber\\
& \leq {1 \over n} I(\bX{1} , \bX{2}; \bY \vert \Bt) + \epsilon_{n} \nonumber\\
& \leq {1 \over n} I(\bX{1} , \bX{2}, \bX{r}; \bY \vert \Bt) + \epsilon_{n} \nonumber\\
& ={1 \over n} \ I(\bX{1},\bX{2},\bX{r} ; g_1e^{j\theta_1}\bX{1} + g_2e^{j\theta_2}\bX{2} + g_re^{j\theta_r}\bX{r} + \bZ)
+ \epsilon_{n} \nonumber\\
& \leq \ {1 \over n} \sum_{i=1}^{n} I(X_{1i},X_{2i},X_{ri} ; g_1e^{j\theta_1}X_{1i} + g_2e^{j\theta_2}X_{2i}
+ g_re^{j\theta_r}X_{ri}+Z_{i}) + \epsilon_{n} \nonumber\\
& \leq \ I(X_{1},X_{2},X_{r} ; g_1e^{j\theta_1}X_{1}
+ g_2e^{j\theta_2}X_{2} + g_re^{j\theta_r}X_{r}+Z) + \epsilon_{n}, \label{HUV_MARC_last}
\end{align}
\noindent where the last step follows with the same RVs as in
\eqref{equation_new_RV_1r}-\eqref{equation_new_RV_1r_last}.

The constraints defined by \eqref{HU_MARC_last} and \eqref{HUV_MARC_last} is an outer bound on
$C_{n}(\Bt)$. But since it applies for a fixed $\Bt$, it is also true for all choices of $\Bt$. By
taking intersection over all values of $\Bt$ and letting $n \rightarrow \infty$, the lemma is
proved.
\end{proof}

To prove the converse part of Theorem \ref{Theorem_Capacity_MARC}, we note by Lemma
\ref{Lemma_Converse_MARC} that each of the bounds of Lemma \ref{Lemma_Converse_MARC} are
simultaneously 
maximized by independent Gaussians. The proof of the converse is complete.

\begin{remark}
Note that to prove the converse part of the Theorem \ref{Theorem_Capacity_MARC_noncausal}, we do
not need the receiver to know the CSI $\Bt$. This is indeed true for other separation theorems of
the paper as well. \thmend
\end{remark}

\subsection{Achievability}\label{Achievability_section_UNCCMARC}


We now establish the same region as achievable for the PI-UNCC-MARC with non-causal cooperation
between the encoders. To derive the achievable region, we perform separate source-channel coding.
The source coding is performed by Slepian-Wolf coding and the channel coding argument is based on
regular block Markov encoding in conjunction with backward decoding \cite{Kramer_MARC_Main1:2004}.
Both source coding and channel coding schemes are explained as follows.

{\em Source Coding}: Recall that the first encoder has non-causal access to the second source
$\bV$. 
From Slepian-Wolf coding \cite{Slepian_Wolf_main:1973}, 
for asymptotically lossless representation of the source $((\bU,\bV),\bV)$, we should have the
rates $(R_1,R_2)$ satisfying
\begin{align}
R_1 & > H(U \vert V), \nonumber\\
R_1 + R_2 & > H(U,V). \nonumber
\end{align}

The source codes are represented by indices $W_1,W_2$ which are then channel coded before being
transmitted.

{\em Channel Coding}:
An achievable region for the discrete memoryless UC-MARC with $2$ users is given based on
the block Markov coding scheme shown in Table \ref{Table_UNCCMARC} combined with backward decoding. 


\begin{table}
\centering
\begin{tabular}{|c|c|c|c|c|}
  \hline
 Encoder  & Block $1$ & Block $2$ & Block $B$ & Block $B+1$\\
  \hline       \hline
    $1$ & $\bx{1}(1,W_{11},W_{21},1)$ & $\bx{1}(W_{11},W_{12},W_{22},W_{21})$ & $\bx{1}(W_{1(B-1)},W_{1B},W_{2B},W_{2(B-1)})$
&  $\bx{1}(W_{1B},1,1,W_{2B})$ \\
  \hline       \hline
    $2$ & $\bx{2}(1,W_{21})$ & $\bx{2}(W_{21},W_{22})$ & $\bx{2}(W_{2(B-1)},W_{2B})$ & $\bx{2}(W_{2B},1)$ \\
 \hline          \hline
$r$ & $\bx{r}(1,1)$ & $\bx{r}(W_{11},W_{21})$ & $\bx{r}(W_{1(B-1)},W_{2(B-1)})$ & $\bx{r}(W_{1B},W_{2B})$ \\
 \hline
\end{tabular}
\caption{Block Markov encoding scheme for UNCC-MARC.}\label{Table_UNCCMARC}
\end{table}



First fix a distribution $p(x_1)p(x_2)p({x_r})$ and construct random codewords
$\bx{1},\bx{2},\bx{r}$ based on the corresponding distributions. The message $W_{i}$ of each
encoder is divided to $B$ blocks $W_{i1},W_{i2},\cdots,W_{iB}$ of $2^{nR_{i}}$ bits each, $i =
1,2$. The codewords are transmitted in $B+1$ blocks based on the block Markov encoding scheme
depicted in Table \ref{Table_UNCCMARC}. Using its non-causal knowledge of the second source,
transmitter $1$ sends the information using the codeword
$\bx{1}(W_{1(t-1)},W_{1t},W_{2t},W_{2{(t-1)}})$, while transmitter $2$ uses codeword
$\bx{2}(W_{2(t-1)},W_{2{t}})$ and the relay sends the codeword $\bx{r}(W_{1(t-1)},W_{2{(t-1)}})$.
We let $B \rightarrow \infty$ to approach the original rates $R_1,R_2$.

At the end of each block $b$, the relay decodes $W_{1b}, W_{2b}$, referred to as forward decoding
\cite{Cover:1979}. Indeed, at the end of the first block, the relay decodes $W_{11},W_{21}$ from
the received signal $\bY_{r}(W_{1b}, W_{2b})$. In the second block, nodes $1$ and $2$ transmit
$\bx{1}(W_{11},W_{12},W_{22},W_{21})$ and $\bx{2}(W_{21},W_{22})$, respectively. The relay decodes
$W_{12},W_{22}$, using the knowledge of $W_{11},W_{21}$, and this is continued until the last
block. Using random coding arguments and forward decoding from the first block, for reliable
decoding of messages $W_{1(b-1)},W_{2(b-1)}$ at the relay after the $b$th block, when $n
\rightarrow \infty$, it is sufficient to have
\begin{align}\label{MARC_achievable_relay_decoding_1}
R_1 & <I(X_1;Y_r \vert X_{2},X_{r},\Bt), \\
R_1+R_2 & < I(X_1,X_2;Y_r \vert X_{r},\Bt).\label{MARC_achievable_relay_decoding_2}
\end{align}


The decoding at the destination, however, is performed based on {\em backward} decoding
\cite{Kramer_MARC_Main2:2004}, \cite{Willems:1982}, i.e., starting from the last block back to the
former ones. As depicted in Table \ref{Table_UNCCMARC}, at the end of block $B+1$, the receiver can
decode $W_{1B},W_{2B}$. Afterwards, by using the knowledge of $W_{1B},W_{2B}$, the receiver goes
one block backwards and decodes $W_{1(B-1)}, W_{2{(B-1)}}$. This process is continued until the
receiver decodes all of the messages. Thus, by applying regular block Markov encoding and backward
decoding as shown in Table \ref{Table_UNCCMARC}, one finds that the destination can decode the
messages reliably if $n \rightarrow \infty$ and
\begin{align}\label{MARC_achievable_destination_decoding_1}
R_1 & < I(X_{1},X_{r};Y \vert X_{2},\Bt), \\
R_1 + R_2 & < I(X_{1},X_{2},X_{r};Y \vert \Bt). \label{MARC_achievable_destination_decoding_2}
\end{align}

The achievability part is complete by first choosing $X_1,X_2$, and $X_r$ as independent Gaussians
and observing that under conditions \eqref{Equation1_Condition_UNCCMARC} and
\eqref{Equation2_Condition_UNCCMARC}, \eqref{MARC_achievable_destination_decoding_1} and
\eqref{MARC_achievable_destination_decoding_2} are tighter bounds than
\eqref{MARC_achievable_relay_decoding_1} and \eqref{MARC_achievable_relay_decoding_2}.

As a result of Theorem \ref{Theorem_Capacity_MARC_noncausal}, in the following corollary, we state
a separation theorem for the PI-UNCC-MAC introduced in Section \ref{Section_introduction_UCC-MAC}
with non-causal cooperation between encoders.

\begin{corollary}\label{Theorem_Capacity_UDMAC_noncausal}
{\em Reliable Communication over a PI-UNCC-MAC}: Necessary conditions for reliable communication of
the source $(U,V)$ over a PI-UNCC-MAC with power constraints $P_1,P_2$ on transmitters, fading
amplitudes $g_1,g_2>0$, and source pair $(\bU,\bV) \sim
{\prod_{i}}p(u_{i},v_{i})$, are given by 
\begin{align}
H(U \vert V) & \leq \log(1+g_1^2P_1/N), \label{separation_1_UDMAC}\\
H(U,V) & \leq \log(1+(g_1^2P_1+g_2^2P_2)/N). \label{separation_2_UDMAC}
\end{align}
\noindent Sufficient conditions for reliable communication are also given by
\eqref{separation_1_UDMAC}-\eqref{separation_2_UDMAC}, with $\leq$ replaced by $<$.\thmend
\end{corollary}
\begin{proof}

The PI-UNCC-MAC is equivalent to a PI-UNCC-MARC where the relay has power constraint $P_r = 0$. As
the relay is thus silent, we may assume without loss that $g_{1r}, g_{2r}$ are arbitrarily large,
and the conditions \eqref{Equation1_Condition_UNCCMARC} and \eqref{Equation2_Condition_UNCCMARC}
are trivially satisfied.
\end{proof}

\section{UC-MARC with causal side information} \label{Section-PI-UCC-MARC}

In this section, we state and prove a separation Theorem for another class of UC-MARC in which the
encoders cooperate causally and by means of a wireless phase fading link between transmitters $1$
and $2$. Unlike the noncausal case discussed in Section \ref{Section_introduction_UNCC-MARC},
$X_{1i}$ is a function of the source signal $\bU$ and {its past received signals
$\bY^{(i-1)}_{1}$}. In the sequel, we sate and prove a separation theorem for the causal PI-UD-MARC
under specific gain conditions.

\begin{theorem}\label{Theorem_Capacity_UDMARC_causal}
{\em Reliable communication over a PI-UCC-MARC}: Consider a PI-UCC-MARC with power constraints
$P_1,P_2,P_r$ on transmitters and the relay, fading amplitudes $g_1,g_2,g_r,g_{1r},g_{2r},g_{21}>0$
as shown in Figure \ref{PI_UCCMARC}, and source pair $(\bU,\bV) \sim {\prod_{i}}p(u_{i},v_{i})$.
Furthermore, assume the gain conditions 
\begin{align}
\label{Equation1_Condition_causal_MARC}
g^{2}_{1r}P_1 & \geq g^{2}_{1}P_1 + g^{2}_{r}P_r, \\
g^{2}_{2r}P_1 & \geq g^{2}_{1}P_1 + g^{2}_{2}P_1 + g^{2}_{r}P_r,  \label{Equation2_Condition_causal_MARC} \\
1 + {{g^{2}_{21}}P_2 \over N}  & \geq 2^{-H(U\vert V)}\left(1 + {g^{2}_{1}P_1 + g^{2}_{2}P_2 +
g^{2}_{r}P_r \over N}  \right).\label{Equation_Condition_causal_extra_UDMARC}
\end{align}

Then, a necessary condition of reliable communication of the correlated sources $(\bU,\bV)$ over
such channel with or without knowledge of $\bm\theta$ at the receiver, is given by 
\begin{align}
H(U \vert V) & \leq \log(1+g_1^2P_1 + g_r^2P_r/N), \label{separation_1_UDMARC_causal}\\
H(U,V) & \leq \log(1+(g_1^2P_1+g_2^2P_2+g_r^2P_r)/N). \label{separation_2_UDMARC_causal}
\end{align}
\noindent Conversely, \eqref{separation_1_UDMARC_causal} and \eqref{separation_2_UDMARC_causal}
also describe sufficient conditions for the causal PI-UCC-MARC with $\leq$ replaced by $<$. \thmend
\end{theorem}

\begin{proof}

{\em Converse}: The proof of the converse part of Theorem \ref{Theorem_Capacity_UDMARC_causal} is
exactly the same as that of Theorem \ref{Theorem_Capacity_MARC_noncausal}, as all of the steps
remain unchanged in the causal setting.

{\em Achievability}: For the achievability part, similar to Section
\ref{Achievability_section_UNCCMARC}, we use separate source and channel coding. We need to show
that given \eqref{separation_1_UDMARC_causal} and \eqref{separation_2_UDMARC_causal}, we can first
losslessly source code the sources to indices $W_1 \in [1,2^{nR_1}],W_2\in [1,2^{nR_2}]$ and then
send $W_1,W_2$ over the channel with arbitrarily small error probability.

\begin{table}
\centering 
\begin{tabular}{|c|c|c|c|c|}
  \hline
 Encoder & Block $1$ & Block $2$ & Block $B$ & Block $B+1$\\
  \hline       \hline
    $1$ & $\bx{1}(1,W_{11},1)$ & $\bx{1}(W_{11},W_{12},W_{21})$ & $\bx{1}(W_{1(B-1)},W_{1B},W_{2(B-1)})$
&  $\bx{1}(W_{1B},1,W_{2B})$ \\
  \hline       \hline
    $2$ & $\bx{2}(1,W_{21})$ & $\bx{2}(W_{21},W_{22})$ & $\bx{2}(W_{2(B-1)},W_{2B})$ & $\bx{2}(W_{2B},1)$ \\
 \hline          \hline
$r$ &  $\bx{r}(1,1)$ & $\bx{r}(W_{11},W_{21})$ & $\bx{r}(W_{1(B-1)},W_{2(B-1)})$ & $\bx{r}(W_{1B},W_{2B})$ \\
 \hline
\end{tabular}
\caption{Block Markov encoding scheme for UCC-MARC.}\label{Table_UCCMARC1}
\end{table}

{\em Source Coding}: Using Slepian-Wolf coding \cite{Slepian_Wolf_main:1973}, for asymptotically
lossless representation of the source $(\bU,\bV)$, we should have the rates $(R_1,R_2)$ satisfying
\begin{align}\label{Source_coding_causal_MARC1}
R_1 & > H(U \vert V),\\
R_2 & > H(V \vert U),\\
R_1 + R_2 & > H(U,V).\label{Source_coding_causal_MARC3}
\end{align}


{\em Channel Coding}: Similar to that given in Section \ref{Achievability_section_UNCCMARC} for the
noncausal PI-UC-MARC, the channel coding argument is again based on block Markov coding with
backward decoding as shown in Table \ref{Table_UCCMARC1}. 
Since $\bV$ is not perfectly and non-causally known to the first encoder, node $1$ needs to first
decode $W_{2t}$ after block $t$ from its received signal over the link between the encoders. 
In order to guarantee correct decoding at the relay and correct backward decoding at the
destination, using standard random coding arguments, the following conditions should be satisfied:

\begin{align}
\label{causal_MARC_achievable_relay_decoding_1}
R_1 & <I(X_1;Y_r \vert X_{2},X_{r},\Bt), \\
R_2 & <I(X_2;Y_r \vert X_{1},X_{r},\Bt), \label{causal_MARC_achievable_relay_decoding_2} \\
R_1+R_2 & < I(X_1,X_2;Y_r \vert X_{r},\Bt),\label{causal_MARC_achievable_relay_decoding_3}
\end{align}
\noindent for decoding at the relay and
\begin{align}
\label{causal_MARC_achievable_destination_decoding_1}
R_1 & < I(X_{1},X_{r};Y \vert X_{2},\Bt), \\
R_1 + R_2 & < I(X_{1},X_{2},X_{r};Y \vert \Bt). \label{causal_MARC_achievable_destination_decoding_2}
\end{align}
\noindent for decoding at the destination respectively.

Additionally, to reliably decode the second encoder's message at the first encoder (which plays the
role of a relay), we need to satisfy the condition
\begin{align}\label{equation_extra_causal}
R_2 < I(X_2;Y_1 \vert X_1,X_r,\Bt).
\end{align}

Computing these conditions for independent Gaussian inputs and using conditions
\eqref{Equation1_Condition_causal_MARC} and \eqref{Equation2_Condition_causal_MARC}, we find the
following achievable region for channel coding:
\begin{align}\label{Achievability_Gaussian_causal_MARC1}
R_{1} &<  \log(1+(g^{2}_{1}P_{1}+g^{2}_{r}P_{r})/N),\\
R_{2} &< \log(1+g^{2}_{21}P_{2}/N),\label{Achievability_Gaussian_causal_MARC2}\\
R_{1}+R_{2} &< \log(1+(g^{2}_{1}P_{1}+g^{2}_{2}P_{2}+g^{2}_{r}P_{r})/N).
\label{Achievability_Gaussian_causal_MARC3}
\end{align}

In order to make the inner bounds of
\eqref{Achievability_Gaussian_causal_MARC1}-\eqref{Achievability_Gaussian_causal_MARC3} coincide
with the outer bounds \eqref{separation_1_UDMARC_causal}, \eqref{separation_2_UDMARC_causal}, we
need to have
\begin{align}
\log(1+(g^{2}_{1}P_{1}+g^{2}_{2}P_{2}+g^{2}_{r}P_{r})/N) - R_1 < \log(1+g^{2}_{21}P_{2}/N), \nonumber
\end{align}
\noindent so that we can drop \eqref{Achievability_Gaussian_causal_MARC2} from the achievability
constraints. But by choosing $R_1 = H(U \vert V) + \epsilon$, with $\epsilon > 0$ arbitrary,
condition \eqref{Equation_Condition_causal_extra_UDMARC} makes
\eqref{Achievability_Gaussian_causal_MARC2} dominated by
\eqref{Achievability_Gaussian_causal_MARC3} for the Gaussian input distributions. Therefore, since
$\epsilon > 0$ is arbitrary, one can easily verify that given \eqref{separation_1_UDMARC_causal}
and \eqref{separation_2_UDMARC_causal} with $\leq$ replaced by $<$, along with the conditions
\eqref{Equation1_Condition_causal_MARC}-\eqref{Equation_Condition_causal_extra_UDMARC}, source and
channel codes of rates $R_1,R_2$ can be found such that
\eqref{Source_coding_causal_MARC1}-\eqref{Source_coding_causal_MARC3}, and
\eqref{causal_MARC_achievable_relay_decoding_1}-\eqref{equation_extra_causal}
simultaneously hold. 
\end{proof}

\begin{corollary}\label{Theorem_Capacity_UDMAC_causal}
{\em  Reliable communication over a PI-UCC-MAC}: Necessary conditions for reliable communication of
the sources $(U,V)$ over the causal PI-UCC-MAC with power constraints $P_1,P_2$ on transmitters,
fading amplitudes $g_1,g_2>0$, and source pair $(\bU,\bV) \sim
{\prod_{i}}p(u_{i},v_{i})$, is given by 

\begin{align}
H(U \vert V) & \leq \log(1+g_1^2P_1/N), \label{separation_1_UDMAC}\\
H(U,V) & \leq \log(1+(g_1^2P_1+g_2^2P_2)/N), \label{separation_2_UDMAC}
\end{align}
\noindent provided

\begin{align}
\label{Equation_Condition_causal_extra_UDMAC}
1 + {{g^{2}_{21}}P_2 \over N}  & \geq 2^{-H(U\vert V)}\left(1 + {g^{2}_{1}P_1 + g^{2}_{2}P_2 \over N}  \right).
\end{align}

Given \eqref{Equation_Condition_causal_extra_UDMAC}, sufficient conditions for reliable
communications are also given by \eqref{separation_1_UDMAC} and \eqref{separation_2_UDMAC}, with
$\leq$ replaced by $<$. \thmend
\end{corollary}
\begin{proof}

The argument is similar to the proof of the Corollary \ref{Theorem_Capacity_UDMAC_noncausal}. The
PI-UCC-MAC is equivalent to a PI-UCC-MARC where the relay has power constraint $P_r = 0$. As the
relay is thus silent, we may assume without loss that $g_{1r}, g_{2r}$ are arbitrarily large. The
conditions \eqref{Equation1_Condition_causal_MARC}-\eqref{Equation2_Condition_causal_MARC} of
Theorem \ref{Theorem_Capacity_UDMARC_causal} with \eqref{Equation_Condition_causal_extra_UDMARC}
being changed to \eqref{Equation_Condition_causal_extra_UDMAC} are then trivially satisfied.
\end{proof}

%

\section{Interference Channel} \label{Section-PI-IC}

We now study the communication of the arbitrarily correlated sources $(U,V)$ over a
phase-asynchronous interference channel introduced in Section \ref{Section_introduction_IC}. The
definition of the joint source-channel code and power constraints are similar to the ones given in
Section \ref{definitions}. However, since there are two decoders in this setup, we define two
indexed decoding functions $g^n_{\bm\theta,1}$ and $g^n_{\bm\theta,2}$ and two error probability
functions
\begin{align}
P_{e1}^n(\Bt_{1}) & = P\{\bU^n \neq \hat{\bU}^n \vert {\bm\theta_{1}}\}
= \sum_{{\bu}^n  \in {\mathcal{U}^n }}  p({\bu}^n) \times P\{\hat{\bU}^n \neq \bu^n \ \vert {\bu}^n , {\bm\theta}\},\nonumber\\
P_{e2}^n(\Bt_{2}) & = P\{\bV^n \neq \hat{\bV}^n \vert {\bm\theta_{2}}\}
= \sum_{({\bv}^n) \in {\mathcal{V}^n}}  p({\bv}^n) \times P\{\hat{\bV}^n \neq \bv^n \ \vert {\bv}^n , {\bm\theta_{2}}\}.
\nonumber
\end{align}
\noindent for each of the corresponding receivers.

Consequently, reliable communications for the PI-IC is defined as:

\begin{definition} We say the source $\{U_i,V_i\}_{i=1}^{n}$ of i.i.d. discrete random variables
with joint probability mass function $p(u,v)$ {\em can be reliably sent} over the PI-IC, if there
exists a sequence of encoding functions
${\mathcal{E}_{n}}\triangleq\{\bbx^{n}_{1}(\bU^{n}),\bbx^{n}_{2}(\bV^{n})\}$
and decoders $g^n_{\bm\theta_{1}},g^n_{\bm\theta_{2}}$ 
such that the output sequence $\bU^{n}$ can be reliably estimated at the first receiver and
$\bV^{n}$ can be reliably estimated at the second receiver over all parameters
$\bm\theta_{1},\Bt_{2}$ respectively. That is,

\begin{align}\label{main_error_probability_IC1}
\left[\sup_{\bm\theta} P_{e1}^n(\bm\theta_{1})\right] \longrightarrow 0, \ \ {\rm as}   \  \ n \rightarrow \infty,\\
\left[\sup_{\bm\theta} P_{e2}^n(\bm\theta_{2})\right] \longrightarrow 0, \ \ {\rm as}   \  \ n \rightarrow \infty.
\label{main_error_probability_IC2}
\end{align}
\thmend
\end{definition}

\begin{theorem}\label{Theorem_Capacity_IC}
{\em Reliable Communications over a PI-IC}: A necessary condition of reliably sending arbitrarily
correlated sources $(U,V)$ over a PI-IC with power constraints $P_1,P_2$ on transmitters, fading
amplitudes $g_{11},g_{12},g_{21},g_{22}>0$, and source pair $(\bU,\bV) \sim
{\prod_{i}}p(u_{i},v_{i})$, with the strong interference condition
\begin{align} \label{Equation_Condition_IC1}
g_{11} &\geq g_{12} \\
g_{22} &\geq g_{21} \label{Equation_Condition_IC2}
\end{align}
\noindent with or without knowledge of $\bm\theta$ at the receiver, is given by
\begin{align}
H(U \vert V) & \leq \log(1+g^{2}_{11}P_1/N), \label{separation_1_IC}\\
H(V \vert U) & \leq \log(1+g^{2}_{22}P_2/N), \label{separation_2_IC}\\
H(U,V) & \leq \min \left\{ \log(1+(g_{11}^{2}P_1+g_{21}^2P_2)/N), \log(1+(g^{2}_{12}P_1+g_{22}^2P_2)/N) \right\}. \label{separation_3_IC}
\end{align}
\noindent The same conditions \eqref{separation_1_IC}-\eqref{separation_3_IC} with $\leq$ replaced
by $<$ describe the achievability region. \thmend
\end{theorem}

\subsection{Converse}

In this section, we derive an outer bound on the capacity region and prove the converse part of
Theorem \ref{Theorem_Capacity_IC} for the interference channel.
\begin{lemma}\label{Converse_lemma_IC}
{\em Converse}: Let $\{\bbx^{n}_{1}(\bu^{n}),\bbx^{n}_{2}(\bv^{n})\}$, and
$g^n_{\bm\theta1},g^n_{\bm\theta 2}$ be sequences in $n$ of codebooks and decoders for the PI-IC
for which \eqref{main_error_probability_IC1} and \eqref{main_error_probability_IC2} hold.
Then we have
\begin{align}
H(U\vert V)& \leq \min_{\Bt_1} I(X_1;e^{j\theta_{11}}X_1+Z), \label{Converse_Lemma_1st_eq_IC}\\
H(V\vert U)& \leq \min_{\Bt_2} I(X_2;e^{j\theta_{22}}X_2+Z), \label{Converse_Lemma_2nd_eq_IC} \\
H(U,V) \leq \min \biggl\{ \min_{\Bt_1} I(X_1,X_2;e^{j\theta_{11}}X_1+g_{21}&e^{j\theta_{21}}X_2+Z) ,
\min_{\Bt_2} I(X_1,X_2;g_{12}e^{j\theta_{12}}X_1+e^{j\theta_{22}}X_2+Z) \biggr\}
\label{Converse_Lemma_3rd_eq_IC}
\end{align}
\noindent for some {\em joint} distribution $p_{X_1,X_2}$ such that $\mathbb{E}\vert X_1\vert^2
\leq P_1, \mathbb{E}\vert X_2\vert^2 \leq P_2$. \thmend
\end{lemma}
\begin{proof}

First, fix a PI-IC with given parameters $(\theta_1,\theta_2)$, a codebook $\mathcal{C}$, and
induced {\em empirical} distribution $p_{\bm\theta}(\bu,\bv,\bx{1},\bx{2},\by_{1},\by_{2})$. Then,
we note that by using the strong interference conditions of \eqref{Equation_Condition_IC1} and
\eqref{Equation_Condition_IC2}, one can argue that both of the receivers can decode both of the
sequences $\bU,\bV$ provided there are encoders and decoders such that each receiver can reliably
decode its {\em own} source sequence (see \cite{Sato:1981} for details). Thus, $\bU,\bV$ can both
be decoded from both $\bY_{1},\bY_{2}$. Thus, we have the intersection of two PI-MACs and the
result follows from Theorem \ref{Theorem_Capacity_MAC}.
\end{proof}

\subsection{Achievability}

The achievability part of Theorem \ref{Theorem_Capacity_IC} can be obtained by noting that if we
make joint source-channel codes such that both receivers are able to decode both messages, then we
will have an achievable region. Thus, the interference channel will be divided to two PI-MACs and
the achievable region will be again the intersection of the achievable regions of the two PI-MACs
as given in Theorem \ref{Theorem_Capacity_MAC}.

\section{Interference Relay Channel (IRC)} \label{Section-PI-IRC}

In this section, we prove a separation theorem for the PI-IRC {introduced in Section
\ref{Section_introduction_IRC}} under some non-trivial constraints on the channel gains which can
be considered as a {\em strong interference} situation for the IRC. The definitions of reliable
communication and joint source-channel codes for the PI-IRC are similar to those for the PI-IC. We
first state the separation theorem and consequently give the proofs of the converse and
achievability parts.

\begin{theorem}\label{Theorem_Capacity_IRC} {\em  Reliable communication over a PI-IRC}: Consider a PI-IRC with power constraints $P_1,P_2,P_r$ on
transmitters, fading amplitudes $g_{11},$ $g_{21},$ $g_{12},$ $g_{22} \geq 0$ between the
transmitters and the receivers, $g_{r1},g_{r2} \geq 0$ between the relay and the receivers, and
$g_{1r},g_{2r}>0$ between the transmitters and the relay. Assume also that the network operates
under the gain conditions
\begin{align}
&{g_{11} \over g_{12}}  = {g_{r1} \over g_{r2}} = \alpha < 1, \label{Equation_Condition_IRC3}\\
&g^{2}_{11}P_1 + g^{2}_{r1}P_r \leq g^{2}_{1r}P_1, \label{Equation_Condition_IRC1}\\
&g^{2}_{22}P_2 + g^{2}_{r2}P_r \leq g^{2}_{2r}P_2, \label{Equation_Condition_IRC2}\\
& (1-\alpha^2) \ g^{2}_{12}P_{1} \leq \alpha^{2} g^{2}_{r2} P_{r},  \label{Equation_Condition_IRC4}\\
&{\left({1-\alpha^2}\right){g_{12}^{2}}P_1 \over P_2} + {\left({1-\alpha^2}\right){g_{r2}^{2}}P_r \over P_2} + {g_{22}^{2}} \leq {g_{21}^{2}}.
\label{Equation_Condition_IRC5}
\end{align}

Then, a necessary condition for reliably sending a source pair $(\bU,\bV) \sim
{\prod_{i}}p(u_{i},v_{i})$, over such PI-IRC is given by

\begin{align}
H(U \vert V) & \leq \log(1+(g_{11}^2P_1+g_{r1}^2P_r)/N), \label{separation_1_IRC}\\
H(V \vert U) & \leq \log(1+(g_{22}^2P_2+g_{r2}^2P_r)/N), \\
H(U,V) & \leq \log(1+(g_{12}^2P_1+g_{22}^2P_2+g_{r2}^2P_r)/N). \label{separation_3_IRC}
\end{align}
\noindent Moreover, a sufficient condition for reliable communication is also given by
\eqref{separation_1_IRC}-\eqref{separation_3_IRC}, with $\leq$ replaced by $<$, when $\Bt$ is known
at the receivers. \thmend
\end{theorem}
The proof of Theorem \ref{Theorem_Capacity_IRC} is discussed in the two following subsections.
First, the converse is proved and afterwards, we prove the achievability part of Theorem
\ref{Theorem_Capacity_IRC}.

\subsection{Converse}

\begin{lemma}\label{Converse_lemma_IRC}
{\em PI-IRC Converse}: Let $\mathcal{E}_{n}$ be a sequence in $n$ of encoders, and
$g^n_{1\Bt},g^n_{2\Bt}$ be sequences in $n$ of decoders for the PI-IRC for which $ \sup_{\bm\theta}
P_{e1}^n(\bm\theta), P_{e2}^n(\bm\theta)\longrightarrow 0$, as $n \rightarrow \infty$, then we have
\begin{align}
H(U\vert V) &\leq \min_{\Bt \in {\bf\Phi}_{c}} I(X_1,X_r;g_{11}e^{j\theta_{11}}X_1+g_{r1}e^{j\theta_{r1}}X_r+Z), \label{Lemma_Converse_1st_eq_IRC}\\
H(V\vert U) &\leq \min_{\Bt \in {\bf\Phi}_{c}} I(X_2,X_r;g_{22}e^{j\theta_{22}}X_2+g_{r2}e^{j\theta_{r2}}X_r+Z), \label{Lemma_Converse_2nd_eq_IRC} \\
H(U,V) &\leq \min_{\Bt \in {\bf\Phi}_{c}} I(X_1,X_2,X_r;g_{12}e^{j\theta_{12}}X_1+g_{22}e^{j\theta_{22}}X_2+g_{r2}
e^{j\theta_{r2}}X_{r}+Z), \label{Lemma_Convers_3rd_eq_IRC}
\end{align}
\noindent for some {\em joint} distribution $p_{X_1,X_2,X_r}$ such that $\mathbb{E}\vert X_1\vert^2
\leq P_1, \mathbb{E}\vert X_2\vert^2 \leq P_2, \mathbb{E}\vert X_r\vert^2 \leq P_r$, where
${\bf\Phi}_{c} \triangleq \{\Bt: \theta_{11}=\theta_{12}, \theta_{r1}=\theta_{r2}\}$. \thmend
\end{lemma}

\begin{proof}

First, fix a PI-IRC with given parameter $\Bt \in {\bf\Phi}_{c}$, a codebook $\mathcal{C}$, and
induced {\em empirical} distribution $p_{\bm\theta}(\bu,\bv,\bx{1},\bx{2},\bx{r},\by_{1},\by_{2})$.
Since for this fixed choice of ${\bm\theta}$, $P^n_{e1}(\bm\theta), P^n_{e2}(\bm\theta) \rightarrow
0$, from Fano's inequality, we have
\begin{align}
{1 \over n}H(\bU \vert \bY_{1},\Bt) & \leq {1 \over n}{P_{e1}^n(\bm\theta)} \log \|\mathcal{U}^{n}\| + {1 \over n}
 \triangleq \epsilon_{1n}(\Bt),\nonumber\\
{1 \over n}H(\bV \vert \bY_{2},\Bt) & \leq {1 \over n}{P_{e2}^n(\bm\theta)} \log \|\mathcal{V}^{n}\| + {1 \over n}
\triangleq \epsilon_{2n}(\Bt),\nonumber
\end{align}
and $\epsilon_{1n}(\Bt),\epsilon_{2n}(\Bt) \rightarrow 0$, where convergence is uniform in
$\bm{\theta}$. Defining $\sup_{\Bt}\epsilon_{in}(\Bt) = \epsilon_{in}, i=1,2$ and following similar
steps as those resulting in \eqref{Equation_1st_outer_main_MARC}, we have

\begin{align}\label{Equation_1st_outer_main_IRC}
H(U \vert V) & \leq {1 \over n}I(\bX{1} , \bX{r}; \bY_{1} \vert \bV, \bX{2}, \Bt) + \epsilon_{1n}, \\
H(V \vert U) & \leq {1 \over n}I(\bX{2},\bX{r} ; \bY_{2} \vert \bU, \bX{1}, \Bt) + \epsilon_{2n}. \label{Equation_2nd_outer_main_IRC}
\end{align}


\noindent As in Section \ref{Section_Converse_MARC}, we can upper bound
\eqref{Equation_1st_outer_main_IRC}, \eqref{Equation_2nd_outer_main_IRC} and derive
\eqref{Lemma_Converse_1st_eq_IRC} and \eqref{Lemma_Converse_2nd_eq_IRC}. Next, to derive
\eqref{Lemma_Convers_3rd_eq_IRC}, we define a random vector $\tilde{\bZ}_{1} \sim
\mathcal{C}\mathcal{N}(0,(1-\alpha)N\bf{I})$ with $\bf{I}$ the $n \times n$ identity matrix, and
bound $H(U,V)$ as follows:
\begin{align}
H(U,V) & = {1 \over n} H(\bU,\bV) \nonumber\\
& = {1 \over n} H(\bV) + {1 \over n} H(\bU \vert \bV) \nonumber\\
& = {1 \over n} H(\bV) + {1 \over n} H(\bU \vert \bV, \bX{2}) \nonumber\\
& = {1 \over n} I(\bV ; \bY_{2} \vert \Bt) + {1 \over n} I(\bU;\bY_{1} \vert \bV, \bX{2},\Bt) + {1 \over n} H(\bV \vert \bY_{2},\Bt)
+ {1 \over n}H(\bU \vert \bV,\bX{2},\bY_{1},\Bt) \nonumber \\
& \leq {1 \over n} I(\bX{2} ; \bY_{2}\vert \Bt) + {1 \over n} I(\bX{1};\bY_{1} \vert \bV, \bX{2},\Bt) + \epsilon_{1n} + \epsilon_{2n} \nonumber\\
& \leq {1 \over n} I(\bX{2} ; \bY_{2}\vert \Bt) + {1 \over n} I(\bX{1}, \bX{r};\bY_{1} \vert \bV, \bX{2},\Bt) + \epsilon_{1n} + \epsilon_{2n}
\nonumber\\
& \leq {1 \over n} I(\bX{2} ; \bY_{2}\vert \Bt) + {1 \over n} \left[h(\bY_{1} \vert \bX{2},\Bt) - h(\bZ_{1})\right] +
\epsilon_{1n} + \epsilon_{2n} \nonumber \\
& = {1 \over n} I(\bX{2} ; \bY_{2}\vert \Bt) + {1 \over n} I(\bX{1}, \bX{r};\bY_{1} \vert \bX{2}, \Bt) + \epsilon_{1n} +
\epsilon_{2n} \nonumber \\
& \leq {1 \over n} I(\bX{2} ; \bY_{2}\vert \Bt) + {1 \over n} I(\bX{1}, \bX{r};g_{11}e^{j\theta_{11}}\bX{1}+
g_{r1}e^{j\theta_{r1}}\bX{r} + \bZ_{1} \vert \bX{2}) + \epsilon_{1n} + \epsilon_{2n} \label{IRC_outer_bound_before_noise_div}\\
& \stackrel{(\rm a)}{=} {1 \over n} I(\bX{2} ; \bY_{2}) + {1 \over n} I(\bX{1}, \bX{r};g_{11}e^{j\theta_{11}}\bX{1}+
g_{r1}e^{j\theta_{r1}}\bX{r} + \alpha\bZ_{1} + {{\tilde{\bZ}}_{1}}\vert \bX{2}) + \epsilon_{1n} + \epsilon_{2n} \nonumber\\
& \stackrel{(\rm b)}{=} {1 \over n} I(\bX{2} ; \bY_{2}) + {1 \over n} I(\bX{1}, \bX{r};g_{11}e^{j\theta_{11}}\bX{1}+
g_{r1}e^{j\theta_{r1}}\bX{r} + \alpha\bZ_{2} + {{\tilde{\bZ}}_{1}}\vert \bX{2}) + \epsilon_{1n} + \epsilon_{2n}\label{IRC_outer_bound_replace_noises}\\
& \stackrel{(\rm c)}{=} {1 \over n} I(\bX{2} ; \bY_{2}) + {1 \over n} I(\bX{1}, \bX{r};\alpha g_{12}e^{j\theta_{12}}\bX{1}+
\alpha g_{r2}e^{j\theta_{r2}}\bX{r} + \alpha\bZ_{2} + {{\tilde{\bZ}}_{1}}\vert \bX{2}) + \epsilon_{1n} + \epsilon_{2n}\nonumber\\
& \stackrel{(\rm d)}{=} {1 \over n} I(\bX{2} ; \bY_{2}) + {1 \over n} I(\bX{1}, \bX{r};\alpha g_{12}e^{j\theta_{12}}\bX{1}+
\alpha g_{r2}e^{j\theta_{r2}}\bX{r} + \alpha\bZ_{2} \vert \bX{2}) + \epsilon_{1n} + \epsilon_{2n}\nonumber\\
& = {1 \over n} I(\bX{2} ; \bY_{2}) + {1 \over n} I(\bX{1}, \bX{r};
\alpha \bY_{2} \vert \bX{2}) + \epsilon_{1n} + \epsilon_{2n} \nonumber\\
& \stackrel{(\rm e)}{=} {1 \over n} I(\bX{2} ; \bY_{2}) + {1 \over n} I(\bX{1}, \bX{r};
\bY_{2} \vert \bX{2}) + \epsilon_{1n} + \epsilon_{2n} \nonumber\\
& = {1 \over n} I(\bX{1},\bX{2},\bX{r} ; \bY_{2}) + \epsilon_{1n} + \epsilon_{2n},
\end{align}
\noindent where $(\rm a), (\rm b)$ follows from the fact that by preserving the noise marginal
distribution, the mutual information does not change. The noise term $\bZ_{1}$ in
\eqref{IRC_outer_bound_before_noise_div} is thus divided into two independent terms
$\alpha{\bZ_{1}} + \tilde{\bZ}_{1}$, and then $\bZ_{1}$ is replaced by $\bZ_{2}$ to obtain
\eqref{IRC_outer_bound_replace_noises}. Also, $(\rm c)$ follows from
\eqref{Equation_Condition_IRC3} and the fact that in ${\bf\Phi}_{c}$, $\theta_{11} = \theta_{12}$
and $\theta_{r1} = \theta_{r2}$, $(\rm d)$ follows since reducing the noise may only increase the
mutual information, and $(\rm e)$ follows from the fact that linear transformation
does not change mutual information: $I(X;Y) = I(X;\alpha Y)$. 

We can now further upper bound $H(U,V)$ by the fact that the upper bound $\mathcal{I}(\Bt)$ is true
for all values of $\Bt \in {\bf\Phi}_{c}$:
\begin{align}
H(U,V) & \leq \min_{\Bt \in {\bf\Phi}_{c}} \left\{{1 \over n} I(\bX{1},\bX{2},\bX{r} ; \bY_{2})\right\} + \epsilon_{1n} + \epsilon_{2n} \nonumber\\
& \leq \min_{\Bt \in {\bf\Phi}_{c}} \left\{ {1 \over n} \sum_{i=1}^{n} I(X_{1i},X_{2i},X_{ri};Y_{2i}) \right\} + \epsilon_{1n} + \epsilon_{2n}
\nonumber\\
& \stackrel{(\rm a)}{=} \min_{\Bt \in {\bf\Phi}_{c}} I(X_{1},X_{2},X_{r};Y_{2} \vert W) + \epsilon_{1n} + \epsilon_{2n} \nonumber\\
& {=} \min_{\Bt \in {\bf\Phi}_{c}} \left[ h(Y_{2}\vert W) - h(Z) \right] + \epsilon_{1n} + \epsilon_{2n} \nonumber\\
& {\leq} \min_{\Bt \in {\bf\Phi}_{c}} \left[ h(Y_{2}) - h(Z) \right] + \epsilon_{1n} + \epsilon_{2n} \nonumber\\
& {=} \min_{\Bt \in {\bf\Phi}_{c}}  I(X_{1},X_{2},X_{r};Y_{2}) + \epsilon_{1n} +
\epsilon_{2n} \nonumber
\end{align}

\noindent where $(\rm a)$ follows by defining the time-sharing RV $W$ and RVs $X_{1}, X_{2}, X_{r}$
as in \eqref{equation_new_RV_1r}-\eqref{equation_new_RV_1r_last} with the power constraints similar
to \eqref{single_power_constraint}. By letting $n \rightarrow \infty$, the proof of the lemma is
complete.
\end{proof}

Using the key lemma, we maximize the upper bounds of Lemma \ref{Converse_lemma_IRC} with the
independent Gaussians and the proof of the converse part is complete.

\subsection{Achievability}

The achievability part is again proved by separate source-channel coding:

\noindent {\em Source Coding}: Using Slepian-Wolf coding, the source $(\bU,\bV)$ is source coded,
requiring the rates $(R_1,R_2)$ to satisfy
\eqref{Source_coding_causal_MARC1}-\eqref{Source_coding_causal_MARC3}.

{\em Channel Coding}: Using the block Markov coding shown in Table \ref{Table_UCIRC} in conjunction
with backward decoding at the receivers (note: both receivers decode all messages) and forward
decoding at the relay, we derive the following necessary conditions to find reliable channel codes
for a compound IRC with $2$ transmitters and a relay $r$:
\begin{align}\label{eq1_Lemma_achievability_IRC}
R_{1} &< \min \left\{ I(X_1;Y_r \vert X_{2},X_{r},\Bt), I(X_{1},X_{r};Y_{1} \vert X_{2},\Bt), I(X_{1},X_{r};Y_{2} \vert X_{2}, \Bt) \right\},\\
R_{2} &< \min \left\{ I(X_2;Y_r \vert X_{1},X_{r}, \Bt), I(X_{2},X_{r};Y_{1} \vert X_{1}, \Bt), I(X_{2},X_{r};Y_{2} \vert X_{1}, \Bt) \right\},
\label{eq2_Lemma_achievability_IRC}\\
R_{1}+R_{2} &<  \min  \left\{I(X_1,X_2;Y_r \vert X_{r}, \Bt),
I(X_{1},X_{2},X_{r};Y_{1} \vert \Bt), I(X_{1},X_{2},X_{r};Y_{2}\vert \Bt)\right\},\label{eq3_Lemma_achievability_IRC}
\end{align}
\noindent for some input distribution $p(x_1)p(x_2)p({x_r})$.

\begin{table}
\centering
\begin{tabular}{|c|c|c|c|c|}
  \hline
 Encoder  & Block $1$ & Block $2$ & Block $B$ & Block $B+1$\\
  \hline       \hline
    $1$ & $\bx{1}(1,W_{11})$ & $\bx{1}(W_{11},W_{12})$ & $\bx{1}(W_{1(B-1)},W_{1B})$ &  $\bx{1}(W_{1B},1)$ \\
  \hline       \hline
    $2$ & $\bx{2}(1,W_{21})$ & $\bx{2}(W_{21},W_{22})$ & $\bx{2}(W_{2(B-1)},W_{2B})$ & $\bx{2}(W_{2B},1)$ \\
 \hline          \hline
$r$ & $\bx{r}(1,1)$ & $\bx{r}(W_{11},W_{21})$ & $\bx{r}(W_{1(B-1)},W_{2(B-1)})$ & $\bx{r}(W_{1B},W_{2B})$ \\
 \hline
\end{tabular}
\caption{Block Markov encoding scheme for IRC.}\label{Table_UCIRC}
\end{table}

Computing the mutual informations in
\eqref{eq1_Lemma_achievability_IRC}-\eqref{eq3_Lemma_achievability_IRC} for independent Gaussians
$X_1 \sim {\mathcal{C}}{\mathcal{N}}(0,P_1)$, $X_2 \sim {\mathcal{C}}{\mathcal{N}}(0,P_2)$, $X_r
\sim {\mathcal{C}}{\mathcal{N}}(0,P_r)$, we find by \eqref{Equation_Condition_IRC1} and
\eqref{Equation_Condition_IRC5} that
\begin{align}
I(X_1;Y_r \vert X_{2},X_{r}, \Bt) & \geq I(X_{1},X_{r};Y_{1} \vert X_{2}, \Bt), \nonumber\\
I(X_{1},X_{r};Y_{2} \vert X_{2}, \Bt) & \geq I(X_{1},X_{r};Y_{1} \vert X_{2}, \Bt), \nonumber
\end{align}
\noindent respectively, and by \eqref{Equation_Condition_IRC2} and \eqref{Equation_Condition_IRC5}
that
\begin{align}
I(X_2;Y_r \vert X_{1},X_{r}, \Bt) & \geq I(X_{2},X_{r};Y_{2} \vert X_{1},\Bt), \nonumber\\
I(X_{2},X_{r};Y_{1} \vert X_{1},\Bt) & \geq I(X_{2},X_{r};Y_{2} \vert X_{1},\Bt), \nonumber
\end{align}
respectively. Also, the conditions \eqref{Equation_Condition_IRC1}-\eqref{Equation_Condition_IRC4}
together result in
\begin{align}
I(X_1,X_2;Y_r \vert X_{r}, \Bt) \geq I(X_{1},X_{2},X_{r};Y_{2} \vert \Bt), \nonumber
\end{align}
\noindent while the condition \eqref{Equation_Condition_IRC5} makes
\begin{align}
I(X_1,X_2,X_r;Y_1 \vert \Bt) \geq I(X_{1},X_{2},X_{r};Y_{2} \vert \Bt). \nonumber
\end{align}

Hence, due to \eqref{Equation_Condition_IRC1}-\eqref{Equation_Condition_IRC5}, the larger terms
will drop off from the constraints
\eqref{eq1_Lemma_achievability_IRC}-\eqref{eq3_Lemma_achievability_IRC} and we may rewrite the
sufficient conditions as
\begin{align}
R_1 & \leq \log(1+(g_{11}^2P_1+g_{r1}^2P_r)/N), \nonumber\\
R_2 & \leq \log(1+(g_{22}^2P_2+g_{r2}^2P_r)/N), \nonumber \\
R_1+R_2 & \leq \log(1+(g_{12}^2P_1+g_{22}^2P_2+g_{r2}^2P_r)/N). \nonumber
\end{align}

Thus, combining the source coding and channel coding, the achievable region is the same as the
outer bound and the proof of Theorem \ref{Theorem_Capacity_IRC} is complete.

\section{conclusion}\label{Conclusion}
The problem of sending arbitrarily correlated sources over a class of phase asynchronous
multiple-user channels with non-ergodic phase fadings is considered. Necessary and sufficient
conditions for reliable communication are presented and several source-channel separation theorems
are proved by observing the coincidence of both sets of conditions. Namely, outer bounds on the
source entropy content $(H(U\vert V),H(V\vert U),H(U,V))$ are first derived using phase uncertainty
at the encoders, and then are shown to match the achievable regions required by separate
source-channel coding under some restrictions on the channel gains. Although, our results are for
fixed $\Bt$, they are also true for the ergodic case:

\begin{remark}
In all of the above theorems, we assumed that the vector $\Bt$ is fixed over the block length. It
can be shown that the theorems also hold for the ergodic phase fading, i.e., the phase shifts
change from symbol to symbol in an i.i.d. manner, forming a matrix of phase shifts $\bm\Theta$. The
achievability parts of the theorems remain unchanged by the assumption of perfect CSI at the
receiver(s), while in the proofs of the converses, the expectation operation over $\bm\Theta$ is
used instead of taking the minimum (as in \cite{FadiAbdallah_Caire:2008}). One can then use the
results of Remark \ref{remark_key_lemma_ergodic} as a key lemma to prove the optimality of
independent Gaussians for the converse parts. \thmend
\end{remark}

As a result, joint source-channel coding is not necessary under phase incoherence for the networks
studied in this work. We also conjecture that source-channel separation is in fact optimal for all
channel coefficients and not only for the constraints presented in this paper.
%


\bibliographystyle{IEEEtranS}
\bibliography{Hamidreza_bib}

\end{document}